\documentclass{article}
\usepackage{amsmath,amsfonts}
\usepackage{amscd}
\usepackage{graphicx}
\usepackage{color}
\usepackage[latin9]{inputenc}
\usepackage{natbib}
\usepackage{cases}
\usepackage{tikz}
\usepackage{cmll}
\usepackage{authblk}
\usepackage{hyperref}

\newtheorem{theorem}{Theorem}

\newtheorem{lemma}[theorem]{Lemma}

\newtheorem{remark}[theorem]{Remark}

\newtheorem{proof}[theorem]{Proof}
\renewenvironment{proof}[1][Proof]{\noindent\textbf{#1.} }{\ \rule{0.5em}{0.5em}}

\newenvironment{revs}
  {\begin{color}{black} \ignorespaces} 
 {\end{color}}

\begin{document}

\title{Inequality, mobility and the financial accumulation process: A computational economic analysis \thanks{{\bf Accepted for publication in Journal of Economic Interactions and Coordination (DOI: \url{https://doi.org/10.1007/s11403-019-00236-7})}. Authors acknowledge funding for a visiting period from ESCP-Europe. Simone Righi further acknowledges funding from the European Research Council (ERC) under the European Union's Horizon 2020 research and innovation programme (grant agreement No 648693).}}

\author[1]{Yuri Biondi}
\author[2,3]{Simone Righi}
\affil[1]{Cnrs - IRISSO (University Paris Dauphine PSL), Place Mar. Lattre Tassigny 75016 Paris. yuri.biondi@gmail.com}
\affil[2]{University College London, Department of Computer Science, Gower Street 66-72, Office 3.10, WC1E 6EA London (UK). s.righi@ucl.ac.uk}
\affil[3]{MTA TK ``Lend\"{u}let" Research Center for Educational and Network Studies (RECENS), Hungarian Academy of Sciences.}

\renewcommand\Authands{ and }

\maketitle

\begin{abstract}
Our computational economic analysis investigates the relationship between inequality, mobility and the financial accumulation process. Extending the baseline model by Levy et al., we characterise the economic process through stylised return structures generating alternative evolutions of income and wealth through time. First, we explore the limited heuristic contribution of one and two factors models comprising one single stock (capital wealth) and one single flow factor (labour) as pure drivers of income and wealth generation and allocation over time. Second, we introduce heuristic modes of taxation in line with the baseline approach. Our computational economic analysis corroborates that the financial accumulation process featuring compound returns plays a significant role as source of inequality, while institutional arrangements including taxation play a significant role in framing and shaping the aggregate economic process that evolves over socioeconomic space and time.

\vspace{0.7cm}


\noindent 
\textbf{JEL\ classification}:  C46, C63, D31, E02, E21, E27, D63, H22

\end{abstract}



\section{Introduction and Literature Review}

Capital wealth accumulation is an evergreen matter of economic analysis and policy.\footnote{Hereafter, the term ``capital wealth'' combines concepts of capital and wealth to stress the productive nature of wealth considered by our economic analysis. Indeed, we especially point to financial investments, while durable assets held for consumption are excluded from our analysis.} Public finance and financial macroeconomic models define notions of production, income and capital wealth and study their aggregated  evolution over time (\citealt{bertola2006income}; \citealt{Blanchard2011Macro}; \citealt{snowdon2005modern}), as well as their distribution across individuals (see the survey of economic literature by \citealt{sahota1978theories}).  In particular, some notable efforts aim to explain empirical distributions as driven by either certain stochastic processes, or the combined influence of driving factors such as family environment, talent, education and social status. Recent economic modelling strategies include \cite{nirei2007two} and \cite{gabaix2016dynamics}.

Since the fifties, the standard representation of growth denotes a multiplicative process that is also the standard representation for individual financial investments. This modeling strategy implicitly assumes a single capital stock that is measured and reinvested for the aggregate economy over time (\citealt{perroux1949comptes}; \citealt{stone1986nobel}).
Recent advances in dynamic macroeconomic modelling, based upon the representative agent hypothesis, have been criticized for disregarding the aggregate dimension featured by collective and dynamic phenomena \citealt{gallegati1999beyond}). Previously neglected, issues of  income and wealth distributions have gained socioeconomic momentum in the aftermath of the Global Financial Crisis of 2007-08, including through the 99\% movement in US (\citealt{haldane2014unfair}; \citealt{alvaredo2013top}). This movement claims that the increased financialisation of economy and society involves an increased appropriation of income and wealth by the richest 1\% of the population at detriment of the remaining 99\%, leading to more unequal and allegedly unfair distributions of income and wealth. This distributional issue renewed theoretical interest through influential positions taken by leading economists (\citealt{Krugman2013Why,Krugman2014Inequality,Krugman2014InequalityDrag,stiglitz2012price,solow2014rich}) and policy-makers (Haldane 2014)), as well as through the publication of economic history studies conducted by Thomas Piketty and Emmanuel Saez among others, reconstructing long-run statistical time series of income and wealth distributions in US and abroad  (\citealt{atkinson2009top}, \citealt{piketty2014capital},\citealt{piketty2014inequality}).

According to \cite{haldane2014unfair}, ``as ever, dispute rages about the precise statistics. But the long-term patterns are clear enough - and remarkable. Almost half of the growth in US national income between 1975 and 2007 accrued to the top 1\% (\citealt{OECD2014}). In the UK and US, the top 1\%'s share of the income pie has more than doubled since 1980 to around 15\% and their share of the wealth pie has been estimated at up to a third - more than the whole bottom half of the population put together (\citealt{ONS,wolff2012asset}). The five richest households in the UK have greater wealth than the bottom fifth of the population (\citealt{Oxfam})".\footnote{See also \cite{CBO2011}.}

The theoretical issue of income and wealth distributions is well-known since classic economic theorists in the XIX century at least, when the leading economist J.S. \cite{Mill1861} considered ``fair and reasonable that the general policy of the State should favour the diffusion rather than the concentration of wealth.'' At the beginning of the XX century, the leading economist and sociologist V. Pareto argued for the so-called Pareto (power-law) wealth distribution as an empirical regularity \begin{revs}(\citealt{pareto1895legge,pareto1896,pareto1897aggiunta,pareto1971cours})\end{revs}, while the economic statistician C. Gini developed ingenious statistical measurement techniques to capture this inequality through the so-called Gini index (\citealt{gini1912variabilita}). \begin{revs}\footnote{Literature on the Pareto (power-law) distribution of wealth is too vast to be summarized here and outside the purpose and scope of this article, which is not concerned with the statistical form of wealth distribution. Further readings include: \cite{kirman1987pareto,them1990income,druagulescu2001exponential,persky1992retrospectives}} \end{revs}

Recent advances in econophysics point to the functional forms of statistical distributions of income and wealth (\citealt{lux2005emergent}). In particular, some scholars aim to reproduce empirical regularities through simple, elegant additive economic processes (\citealt{angle2006inequality,richmond2001power,solomon2002stable}).
 Other scholars purport to explain the fat tail of these distributions (that is, the tail concerned with the higher ranges of aggregate income and wealth) through multiplicative economic processes which lead to emerging power-laws (\citealt{levy2005market,levy2003investment,Milakovic2003towards}). Some recent contributions suggest the form of a {\it deformed} exponential function, which seems to capture well the empirical regularities of income distribution at the low-middle range, as well as its power-law tail  (\citealt{kaniadakis2001non,kaniadakis2002statistical}). These modelling attempts have raised a lively debate with some economists who were worrying about allegedly poor socioeconomic understanding and lack of theoretical economic underpinnings (\citealt{gallegati2006worrying,lux2005emergent}). Further collaborative and interdisciplinary research has developed the application of the k-deformed exponential function to the parametric modelling of personal income and wealth distributions (\citealt{clementi2015distribution,clementinew,clementi2007kappa,clementi2008kappa,clementi2009kappa,clementi2010model,clementi2012generalized,clementi2012new,clementi2016kappa}). The latter approach provides insights on the drivers of these distributions over time and across the population, while enabling synthetic comparison through inequality and poverty measures that are derived from parametric estimations.

In this context, generalising \cite{champernowne1953model}, \cite{levy2005market} and \citealt{levy2003investment} (Levy et al. thereafter) have developed an elegant modelling strategy purporting to explain the power-law tail of income and wealth distributions under financial market efficiency, and the stochastic distribution of financial returns across individuals active in this market.

In sum, theoretical and societal attention paid to the economic inequality issue raises the question of the cause of this inequality. Whichever tentative response to this question has profound socioeconomic implications, raising further theoretical and applied concerns which go beyond the functional form of statistical distributions of income and wealth across individuals. From this broader perspective, our contribution purports to address two featuring dimensions:
\begin{itemize}
\item[] (i) the inequality of income and wealth allocation across individuals, and its evolution over time;
\item[] (ii) the significance of collective institutional mechanisms including taxation that actively frame and shape this economic process.
 \end{itemize}
In particular, our modelling strategy consists in extending and improving on existing literature by considering these two dimensions. Levy et al. provide a convenient baseline model which subsumes the basic assumptions which characterise widespread economic modelling on these matters.  \cite{fernholz2014instability} and \cite{bertola2006income} review and develop more sophisticated models that maintain similar background assumptions. In this context, Levy et al. have the advantage to reduce the model structure to its minimal, synthetic, and simple formulation. By elaborating on the Levy et al. model, our computational economic analysis will show the relevance of the financial accumulation process that features compound return investment over time. This peculiar accumulation process explains qualitatively both the increasing inequality across individuals, and the decreasing social mobility empirically observed in recent decades. 

The rest of the article is organised as follows. The second section introduces a financial accumulation process model inspired by the Levy et al. model, as baseline scenario. The third section shows the implications of this model for the evolution of inequality and social mobility through time, assessing their sensitivity to changes in variance and non-normal distribution of returns. The fourth section extends the baseline model by introducing decreasing returns and the simple return structure. The comparison with this latter structure corroborates that, without financial accumulation, inequality is not increased over time in the baseline scenario. The fifth section introduces a second flow factor (labour income) along with the stock factor (capital wealth) considered by Levy et al. The introduction of a flow factor may involve an income-saving process that complements and integrates the financial accumulation process driven by inherited wealth. All together, the analysis developed in the first five sections makes clear that distributional effects, which depend on aggregate configurations, have been neglected by the received literature. This preliminary conclusion  paves  the  way  to  introducing  minimal  institutions (\`{a} la Shubik) that denote collective mechanisms related to income and wealth distributions. In particular, the sixth section introduces simple centralised modes of taxation, featuring a proportional taxation model (proportional taxation of periodic net income, uniformly redistributed through provision of universal public service), and a progressive taxation model (progressive taxation of periodic net income, redistributed in a regressive way through direct transfers). A summary of main results concludes.

\section{Modelling strategy for the financial accumulation process}

Levy et al. develop a simple model of aggregate economic process based upon one stock factor (wealth) generating a pure compound rate of return $r_{i,t}$ stochastically distributed across individuals and time periods. This model captures the Pareto law shape in the high-wealth (and high-income) range of the aggregate distributions, where  ``changes in wealth are mainly due to financial investment, and are, therefore, typically multiplicative" (\citealt{levy2005market}, p. 105). This modelling strategy is based on a stochastic multiplicative process of wealth accumulation with lower bound on wealth and homogeneous financial investment talent. According to the authors, this framework implies that ``the only reason for inequality is the stochastic process - chance. This implies that there is no differential ability in asset selection or in timing the [financial] market, which is in line with the efficient-market hypothesis. [...] Homogeneous accumulation talent means that all investors draw their returns randomly from the same distribution (the realized return, however, generally differs from one investor to another)" (\citealt{levy2003investment}, p. 709 and 711).\footnote{In fact, \citealt{levy2005market} (chapter, p. 111, footnote 13) concedes that even joint accumulation processes with heterogeneous accumulation talents are asymptotically Paretian, with the faster-increasing multiplicative process dominating the high-range in the long run. \cite{fernholz2014instability} maintain that, in their model, ``luck alone - in the form of high realised investment returns - [...] creates divergent levels of wealth."}
We formalise the Levy et al. model of financial economic process through the familiar structure of compound return. Thus, wealth $W_{t+1}$ of agent $i$ at time $t+1$ is computed as:

\begin{equation}
W_{i,t+1} = (1 + r_{i,t}) W_{i,t} \,\, \text{ for } r\geq-1
\end{equation}

or 

\begin{equation}
W_{i,T}=W_{i,1} \prod_{t=1}^T (1+r_{i,t})
\label{levyprocess}
\end{equation}

where each individual $i$ draws his actual return $r_{i,t}$ at time $t$ from the same statistical distribution defined as follows: $r_{i,t} \sim N(\mu_r, \sigma_r)$ with $\mu_r, \sigma_r > 0$. We take $\sigma_r$ sufficiently large to enable the possibility of financial investment losses. \begin{revs}Contrary to Levy et. al, our model does not include a reflective lower bound on minimum wealth. This lower boundary would introduce an implicitly redistributive process in the baseline scenario, while we prefer restricting this scenario to pure financial accumulation.  Moreover, Levy et al. require the reflective lower bound in order to obtain the power-law distribution of wealth. The statistical form of wealth distribution is outside the purpose and scope of our article, which focuses instead on the relationship between the financial accumulation process, inequality and social mobility. \end{revs}

In the degenerated case with $r$ constant, the Eq. \ref{levyprocess} becomes the classic formula of compound returns over time:
\begin{equation}
W_{T} = W_{1} (1 + r)^T 
\end{equation}
where for $-1\leq r<0 : W_t \rightarrow_t 0$ and for $r>0 : W_t\rightarrow_t + \infty$.

This stylised model does not pretend to reproduce economic reality in its totality. In particular, it does not introduce consumption, overlapping generations, or windfall gains and losses due to wars or accidents. However, it captures one featuring element of the aggregate economic process: financial accumulation opportunities. Compound returns feature financial investment dynamics and related institutions. Financial institutions, such as investment funds, and widespread measures of financial performance are based upon compound return as reference logic. It seems then particularly significant to disentangle and analyse its impact. The aggregate economic process is increasingly managed through corporate forms that live indefinitely and can then go on performing financial accumulation. On the one hand, financial investment is conducted by institutional investors which are driven by, and assessed against, compound return. On the other hand, eventual redistribution of their financial proceeds is often received by corporate recipients that go on reinvesting those proceeds over time, in a self-referential financial accumulation dynamics.

Throughout all our computational analysis, we assume an initial equal distribution of wealth $W_{i,t=1} = 10 \,\, \forall \,\, i$ across all individuals at initial time $t=1$. This implies that inequality depends entirely on the specifications of the economic process. Furthermore, for sake of simulation, we impose the same random seed to all the various sets of simulations proposed in this article. When not mentioned otherwise, we also  define a population of  $N=5000$, and we run every simulation round for $t_{max}=5000$ steps. \footnote{The Matlab code of the simulations can be found at: \url{https://github.com/simonerighi/BiondiRighi2018_JEIC}} Contrary to Levy et al., we allow the theoretical possibility that individual wealth falls to, and remains at zero level. Individual agents take financial investment risk and may occasionally lose all their capital wealth. \footnote{For simulation purpose, we calibrate the parameter space to make this possibility unlikely. In the scenarios presented in this article, despite the high number of iterations, no agent ever loses its wealth completely. Agents experience partial losses (negative returns), but no complete loss of their wealth.}

Our computational economic analysis disentangles two featuring dimensions to be analysed: wealth inequality across individuals, and social mobility relative to wealth dimension.

Wealth inequality is captured through the Gini Index $G_t$ which summarises the relative concentration of wealth across individual at a certain period of time $t$, defined as follows: 
\begin{equation}
G_t = \left[(N+1)-2\left(\frac{\sum_{k=1}^N (N+1-k)w_{k,t}}{\sum_{k=1}^N w_{k,t}}\right)\right]\frac{1}{N-1} \,\,\, \text{with} \,\,\,  0\leq G_t \leq 1 
\label{gini}
\end{equation}
where $N$ is the number of individuals and $w_{k,t} \leq w_{k+1,t}$ denotes the ranked vector of $W_{i,t}$ at time $t$. Accordingly, $G_t \rightarrow 0$ when individual wealths become more equal, while $G_t \rightarrow 1$ when richer individuals tend to acquire a larger share of aggregate wealth. In order to further corroborate the results obtained observing the Gini Index, we also study other measures of inequality such as the Theil index; the absolute and relative share of income by the top 1\% of the population; and the evolution of the proportion of wealth appropriated by different deciles of wealth. All these measures qualitatively confirm the results, and are thus relegated to supplementary material.

Concerning wealth mobility, our Weighted Mobility Index $M_t$ denotes the relative change in wealth position by agent $i$ between two adjacent time periods $t-1$ and $t$. We consider the average of this index across individuals at each period $t$. Weighted Mobility Index $M_t$ is computed as follows:

\begin{equation}
M_t = \frac{1}{N} \sum_{i=1}^N \left[\frac{\vert Dec\left[W_j\right]_{i,t-1}-Dec\left[W_j\right]_{i,t} \vert }{Dec\left[W_{j=1}\right]_{t} -Dec\left[W_{j=10}\right]_{t}}\right]
\label{WeightedMovementIndex}
\end{equation} 

where $Dec\left[W_j\right]_{i,t}$ represents the median wealth at time $t$ for the decile $j$ in which agent $i$ was at time $t$, while $\left[W_{j=1}\right]_{t}$ and $\left[W_{j=10}\right]_{t}$ denote respectively the median wealth for the first and the last decile at time $t$. This index captures the relative movement of the individual $i$ whenever he moves across deciles, relative to the maximum relative wealth distance between the first and the last decile. By taking its mean for each period across the population $N$, we denote the average wealth-weighted individual capacity to move across deciles period after period. Again, in order to corroborate the results obtained with this indicator, we test other measures of mobility. These additional measures confirm the deductions that can be inferred from $M_t$ and are thus relegated to supplementary material.

\section{The baseline case}

The dynamics of wealth concentration across individuals over time is impressive under the baseline scenario introduced by Levy et al. Wealth distributions become increasingly skewed under various compound return structures where individual returns are extracted from normal and gamma distributions at each period of time. For all these structures, the upper tail of wealth distribution goes on appropriating an increasing share of aggregate wealth over time. This dynamic effect has implications for wealth inequality  (Figure \ref{Fig1B}).  In particular, the Gini Index shows that wealth inequality is magnified under the baseline case, asymptotically tending to its maximal value of one. In particular, drawing upon \cite{fernholz2014instability}' proof, we introduce the following Lemma \ref{prop1} concerned with the evolution of time-average wealth distribution:  

\begin{lemma}
The asymptotic value of a Gini Index based upon time-average wealth tends almost surely to its maximum value of one.
\label{prop1}
\end{lemma}
\begin{proof}
See Appendix \ref{tendtoone}.
\end{proof}

\begin{revs}Following \cite{biondi2018financial}, we further introduce the following lemma concerned with the Gini Index across the population at each point of time:\end{revs}

\begin{revs}
\begin{lemma}
The asymptotic value of the Gini Index $G_t$ on the entire population at a certain point of time $t$ asymptotically tends almost surely to its maximum value of one.
\label{prop2}
\end{lemma}
\begin{proof}
See Proof of Lemma 3.1 in \cite{biondi2018financial}.
\end{proof}
\end{revs}

\begin{figure}[!ht]
\centering
\includegraphics[width=0.48\textwidth]{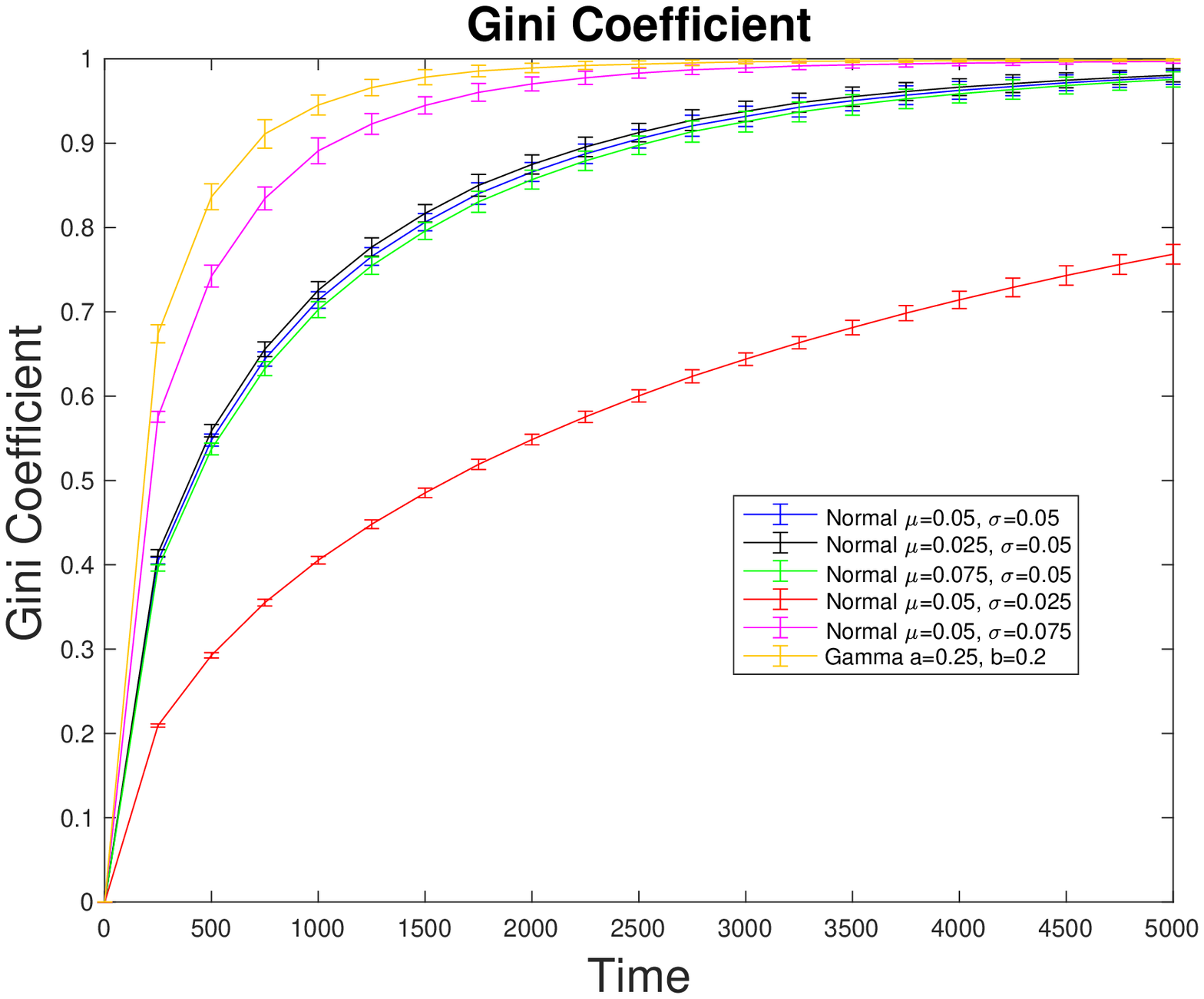}
\includegraphics[width=0.48\textwidth]{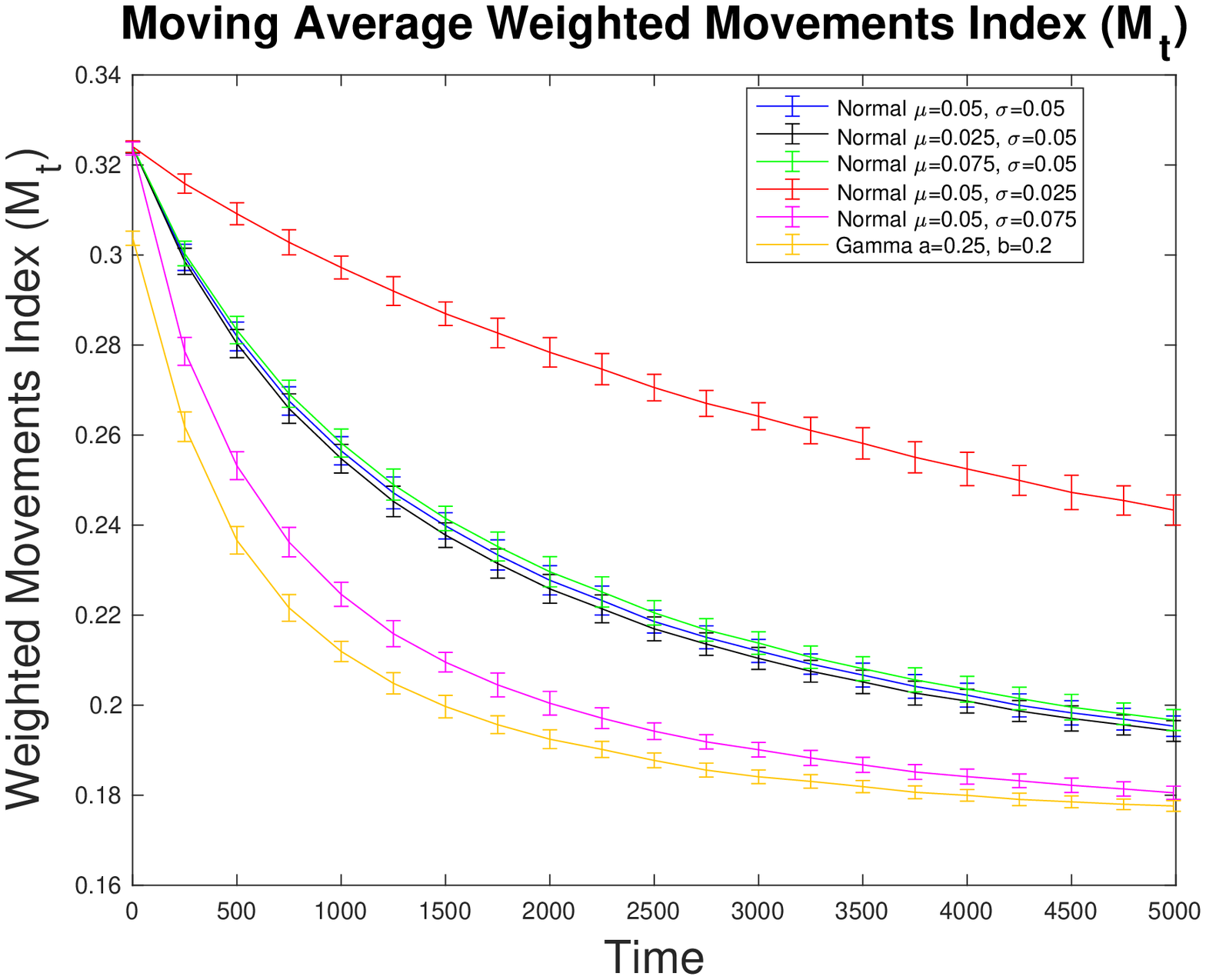}
\caption{Left Panel: Gini Index (Equation \ref{gini}) over time periods for wealth distribution of 5000 individuals under various return structures:  $r_{i,t} \sim N(\mu_r,\sigma_r)$ and $gamma(a,b)$. The Gini Index asymptotically tends to one (see Appendix \ref{tendtoone}). Right Panel: Moving average over ten periods of Weighted Mobility Index defined in Equation \ref{WeightedMovementIndex}. The index is computed under various return structures :  $r_{i,t} \sim N(\mu_r,\sigma_r)$ and $gamma(a,b)$. Averages and standard deviations are computed out of 100 simulations. In all simulations initial wealth $W_{i,t=0}=10 \,\,\, \forall i$}
\label{Fig1B} 
\end{figure}

Furthermore, as shown more systematically by Figure \ref{Fig2B}, the Gini Index is increasing in the variance of the underlying return structure. This concentration effect could be counterbalanced by social mobility relative to wealth, involving the actual capacity of individuals to move across the wealth distribution through time. However, our simulations show that this is not the case for the baseline scenario. Indeed, the Wealth Mobility Index $M_t$ shows that wealth mobility is rapidly decreasing both over time (Figure \ref{Fig1B}, Right Panel) and in the variance of returns (Figure \ref{Fig2B}, Right Panel).

\begin{figure}[!ht]
\centering
\includegraphics[width=0.48\textwidth]{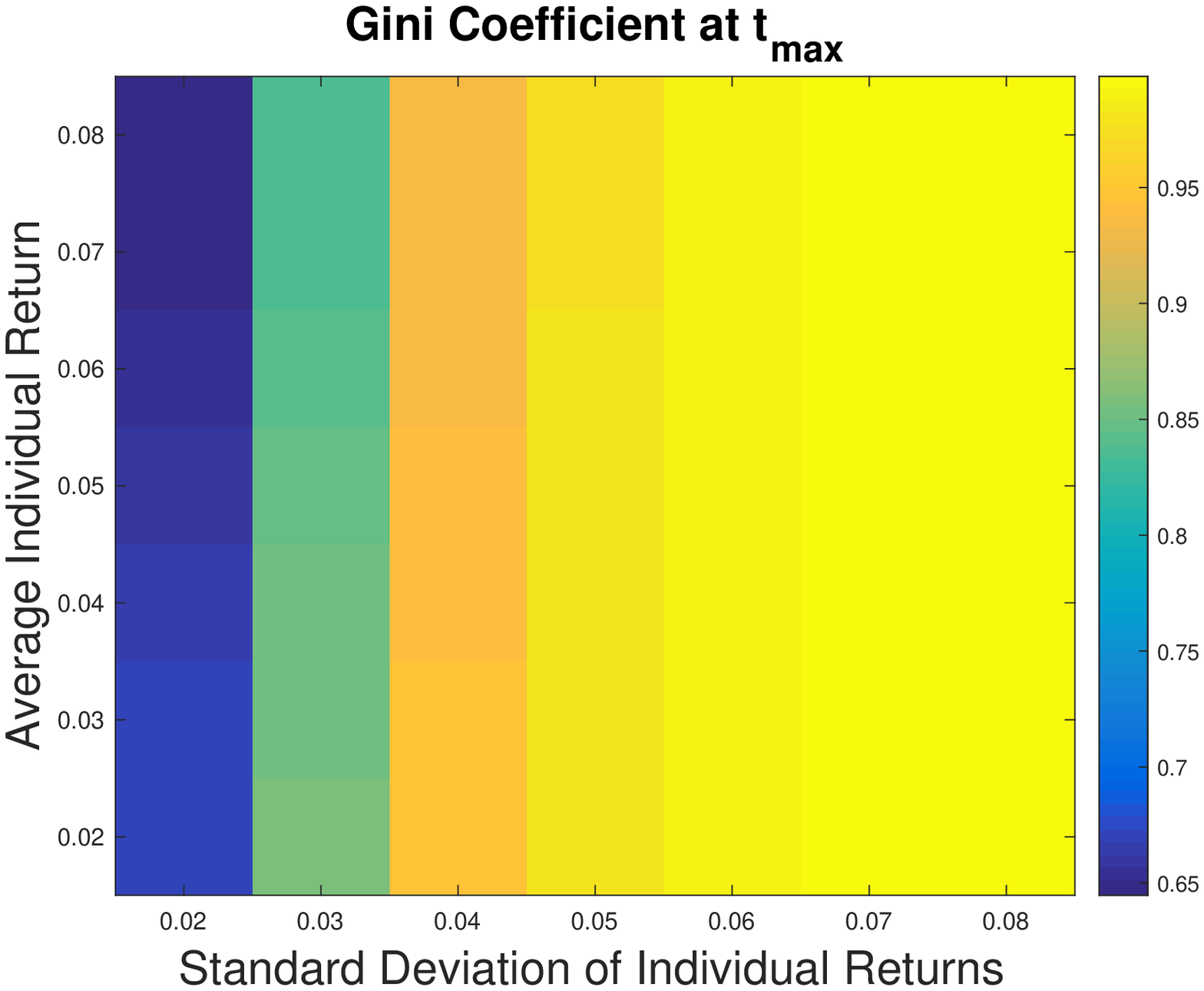}
\includegraphics[width=0.48\textwidth]{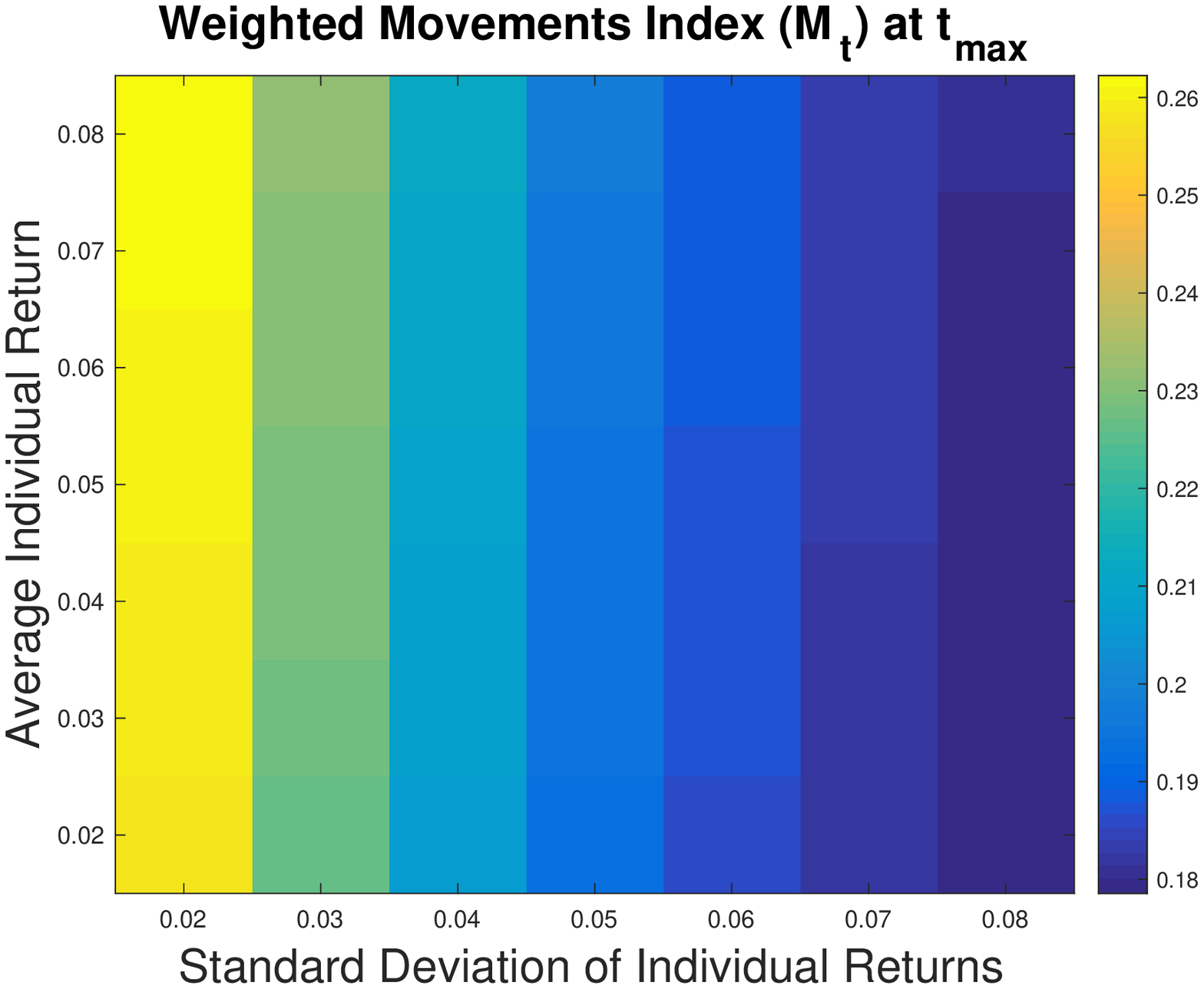}
\caption{Left Panel: Gini Index $G_t$ (Equation \ref{gini}) value at time $t_{max}=5000$ under baseline return structure $r_{i,t} \sim N(\mu_r,\sigma_r)$. Right Panel: Weighted Mobility Index value at time $t_{max}=5000$, as defined in Equation \ref{WeightedMovementIndex}. The index is computed under the baseline return structure :  $r_{i,t} \sim N(\mu_r,\sigma_r)$. Both panels: Values taken by $\mu_r$ are represented on the vertical axis, values taken by $\sigma_r$ are represented on the horizontal axis. In all simulations, initial wealth $W_{i,t=0}=10 \,\,\, \forall i$. 
}
\label{Fig2B} 
\end{figure}

Our result on wealth mobility is further reinforced observing the mobility of the 1\% richest individuals at different points of time. This experiment is visualized in Figure \ref{Fig8}. Having ranked all the agents according to their wealth at one period $t$, the top 1\% is selected. Ranks are normalized by dividing values for the total number of agents so to establish a measure that is independent of population size. In the Left Panel (Figure \ref{Fig8}), we compute the average position of each selected agent (that is, those included in the top 1\% at period $t$) over the following 1000 periods. We then show the distribution of probability of this average rank position. \footnote{The time average of rank positions reduces the impact of idiosyncratic oscillations. Similar results are obtained by replacing the average rank positions with the individual positions after 1000 periods.} The farther is the time period $t$ at which the selection of the top $1\%$ is made, the lower is the average rank position for those top individuals. This implies a decreasing downward mobility and increasing persistence of the top 1\% over time. Moreover,  Figure \ref{Fig8} (Right Panel) shows the evolution of the average rank of individuals in the top 1\% at time $t$ over the next 1000 periods. As time passes, richest individuals at a given period $t$ tend to remain among the richest. Indeed, while individuals that are among the richest at    $t = 10, 100$ may yet revert to some lower position over time (being replaced by previously poorer individuals), individuals that are rich at $t = 1000, 2000$ tend to remain in the top decile of the wealth distribution. Finally, individuals that are rich at $t = 3000, 4000$ tend to remain in the top centile of wealth over time, thus perpetuating their social position relative to wealth.

Both results for wealth inequality and wealth mobility depend especially on returns variance (Figure \ref{Fig2B}). Coeteris paribus, Gini Index is increasing and Wealth Mobility Index is decreasing in return variance  $\sigma_r$ under normally distributed return structures (Figure \ref{Fig2A}, Left Panel). Interestingly, under our baseline case, increased wealth inequality does not depend on mis-alignement between average individual financial returns and aggregate growth. As expected (Figure \ref{Fig2A}, Right Panel), average aggregate growth remains in line with average individual return in our baseline case.

\begin{figure}[!ht]
\centering
\includegraphics[width=0.48\textwidth]{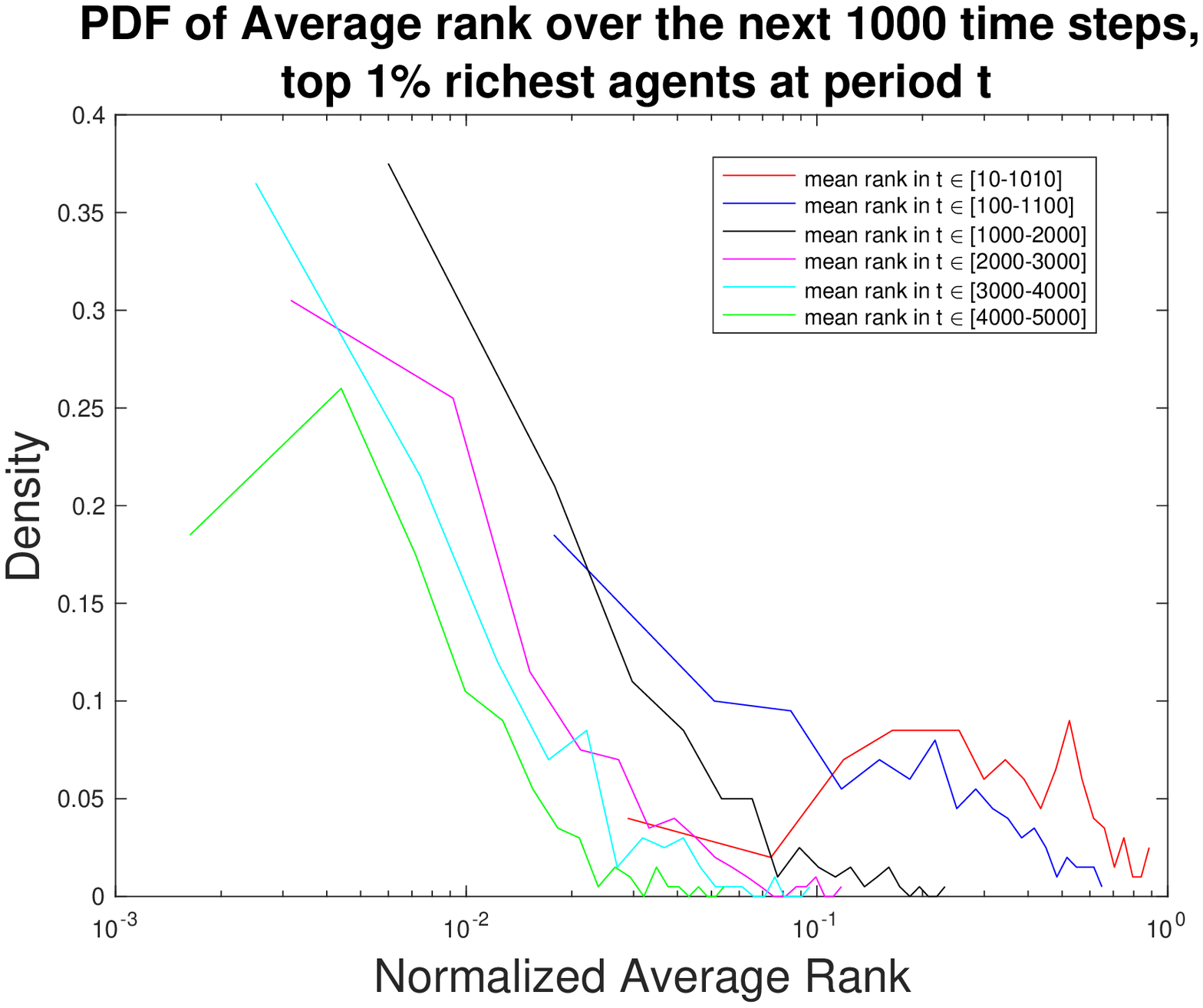}
\includegraphics[width=0.48\textwidth]{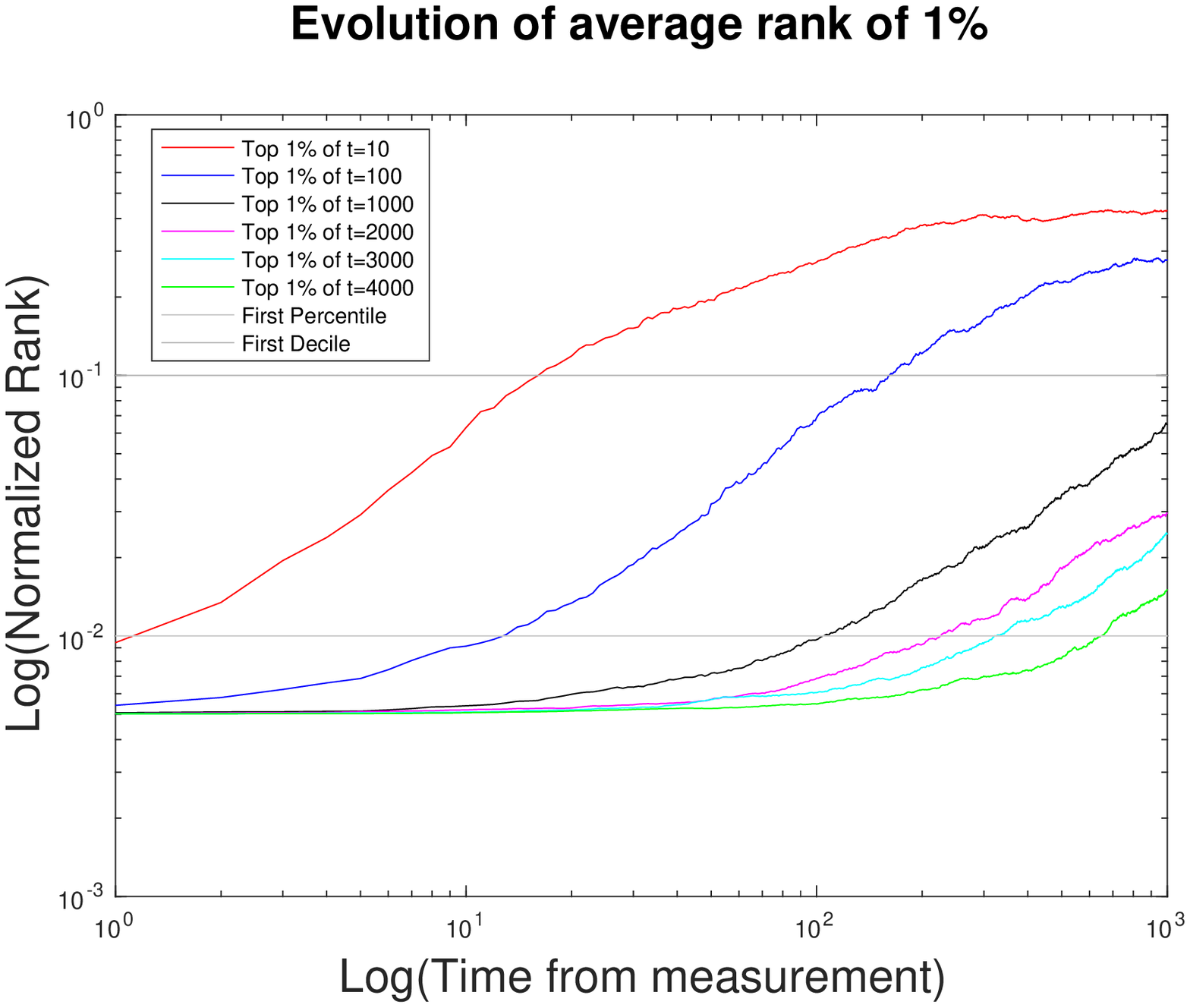}
\caption{This computational analysis captures the chance for agents that are the richest at a certain point of time to remain among the richest throughout further periods. Indicators compare their relative position in the wealth distribution across periods of reference. Left Panel: Distribution of average normalized ranking of agents belonging to the top 1\% at certain periods of time. Right Panel: Evolution of average normalized ranking of agents belonging to the top 1\% at certain periods of time as function of time since that period. For both panels $N=20000$. The simulations are computed under the baseline return structure with $r_{i,t} \sim N(0.05,0.05)$, while initial wealth $W_{i,t=0}=10 \,\,\, \forall i$.}
\label{Fig8} 
\end{figure}

\begin{figure}[!ht]
\centering
\includegraphics[width=0.60\textwidth]{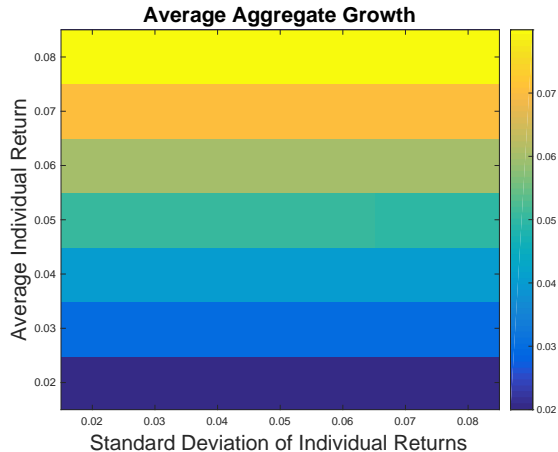}
\caption{Relationship between average aggregate growth and conditions of  individual returns. Aggregate Growth is defined by Equation \ref{Growth} while mean individual return is $\mu_r$  and the standard deviation of individual returns is $\sigma_r$. By assumption, returns are dispersed according to the normal distribution $r_{i,t} \sim N(\mu_r,\sigma_r)$, in the baseline case, while initial wealth $W_{i,t=0}=10 \,\,\, \forall i$. Values taken by $\mu_r\in\left[0.02;0.08\right]$ (with incremental steps of $0.01$) are represented on the vertical axis, while values taken by $\sigma_r\in\left[0.02;0.08\right]$ (with incremental steps of $0.01$) are represented on the horizontal axis.}
\label{Fig2A}
\end{figure}

This sensitivity to return variance in the baseline case is magnified by financial market dynamics. Empirical evidence for financial markets behaviour shows that actual market price and return series are not normally distributed, featuring fat tails and extreme events (\citealt{mandelbrot1963variation,mandelbrot1967variation,mantegna1996introduction,biondirighi2013,biondi2017much}, providing further references). For sake of simulation, we complement normal distribution of returns with a gamma distribution having the same mean as in the baseline case but featuring extreme events, i.e., gamma$(a,b)$ with $a=0.25$ and $b=0.2$ where $a \cdot b=\mu_r=0.05$ and $a \cdot b^2=\sigma_r=0.01$. This parameterization for the gamma distribution is applied to all our simulation analyses, when not stated otherwise. Computational results show that the gamma distribution of returns reinforces wealth concentration and inequality, while undermining wealth mobility over time. In particular, the Gini Index is always superior at each period of time (Figure \ref{Fig2B}, Left Panel), while the Wealth Mobility Index is always inferior (Figure \ref{Fig2B}, Right Panel). This result foreshadows that the non-normality of financial market returns may have an inequality-enhancing impact, favouring skewed accumulation of wealth across individuals and over time.

\section{Decreasing compound returns and simple return structures: History matters}

Levy et al. insist on the stochastic nature of their financial investment process. Our computational economic analysis further points to its cumulative nature over time, depending on the peculiar deployment of compound returns. Along with stochastic extraction of the actual return $r$ for investor $i$ at each time $t$, the financial investment process is further featured by the cumulative impact of the individual series of compound returns on accumulated wealth through historical time (Figure  \ref{Fig1B}). A quick glance at the deterministic reduction of the process model in Equation \ref{levyprocess}  shows that being richer at time t almost assures becoming richer at a further time $t + n$ with $n >> 0$. Coeteris paribus, this evolutionary structure tends to favour investors that become richer earlier in time, that is, investors that accumulate net gains before (and net losses after) the others, since every gain compounds positively, while every loss compounds negatively through time. This cumulative process is exacerbated by the constant mean return to wealth which was assumed in the baseline scenario.  Therefore, rather than `being lucky', this process denotes that `history matters'. This financial accumulation process has important implications for the evolution of wealth through socioeconomic space and time. In a similar vein, \cite{keynes1933economic} would ``trace the beginnings of British foreign investment to the treasure which Drake stole from Spain in 1580", reinvested at annual compound return of 3.25\% over the next centuries to 1930, while remembering its connection to ``avarice and usury and precaution that must be our gods for a little longer still [... to] lead us out of the tunnel of economic necessity into daylight". Hysteresis and path dependency play an important role in explaining the inequality generated by the financial accumulation process. Wealth concentration is ever-increasing over time (Figure \ref{Fig1B}), while relative social mobility is undermined by increasing differences in total wealth\footnote{\cite[p. 7]{levy2003investment} prove that ``the actual wealth distribution converges to a Pareto distribution $\left[\dots\right]$ with minimum wealth, average wealth, and variance that grow over time" when a lower bound on minimum wealth is introduced. Without the latter, the distribution converges to a non-stationary log-normal distribution.}, as showed by the Wealth Mobility Index (Figure \ref{Fig2B}). In sum, idiosyncratic compound returns on investment through time prove to have cumulative effects, making the aggregate distribution of wealth not stationary. In particular, this distribution becomes increasingly right-skewed over time, tending to a limit in which wealth is concentrated entirely at the top (see also \citealt{fernholz2014instability,biondi2018financial}). This right-skewed tail of wealth distribution depends especially on the second order of return distributions at a certain time period $t$ (that is, the variance $\sigma_r$ in case of returns that are normally distributed). 

This section further assesses the relationship between wealth inequality and mobility and the evolution of returns through time. A first extension consists in exploring alternative economic processes characterised by decreasing returns to wealth. For simulation purposes, decreasing returns to aggregate wealth can be introduced by imposing an external constraint on all the returns  $r_{i,t}$ as follows:
\begin{equation}
(1+r_{i,t}) = \frac{1+r_{i,t}}{\log(1+TW_t)} \,\,\,\, \forall t>1
\label{reduced}
\end{equation}
where 
\begin{equation}
TW_t=\sum_{i=1}^N W_{i,t}
\label{tw10}
\end{equation}
Accordingly, all actual returns $r_{i,t}$ $\forall i,t$ decrease in proportion to aggregate wealth, which is positive and increasing on average over time by assumption. Possible and actual net gains and losses are then progressively reduced over time. Under our baseline assumptions, they tend to zero in the long run, that is, $1+r_{i,t} \rightarrow 1$ for $t\rightarrow +\infty$. 
Decreasing returns are relevant here to test the sensitivity of wealth inequality and mobility to time. In addition, positive constant returns are not a reasonable hypothesis asymptotically, where constraints on growth, as well as limits in natural and human resources, typically appear. Decreasing returns may also denote the case of absent or decaying technological development.

Computational results show that wealth inequality is materially reduced, while wealth mobility is improved under decreasing returns. In particular, Gini Index is consistently lower across time, while Weighted Mobility Index, although decreasing, remains asymptotically superior over time, relative to the baseline case (Figure \ref{Fig3B}). Therefore, individuals are always less unequal and more able to move across wealth relative levels in this scenario, relative to the baseline case. Decreasing returns over time progressively reduce the opportunity by individuals to gain or lose from their wealth investment and accumulation. Therefore, individuals do not have sufficient occasions to accumulate wealth over time, both in absolute and relative terms. Decreasing returns reshape both the first and the second order of return distribution, reducing both the total wealth and its dispersion across individuals over time. This result shows that wealth inequality and mobility depend on temporal evolution of returns.

\begin{figure}[!ht]
\centering
\includegraphics[width=0.48\textwidth]{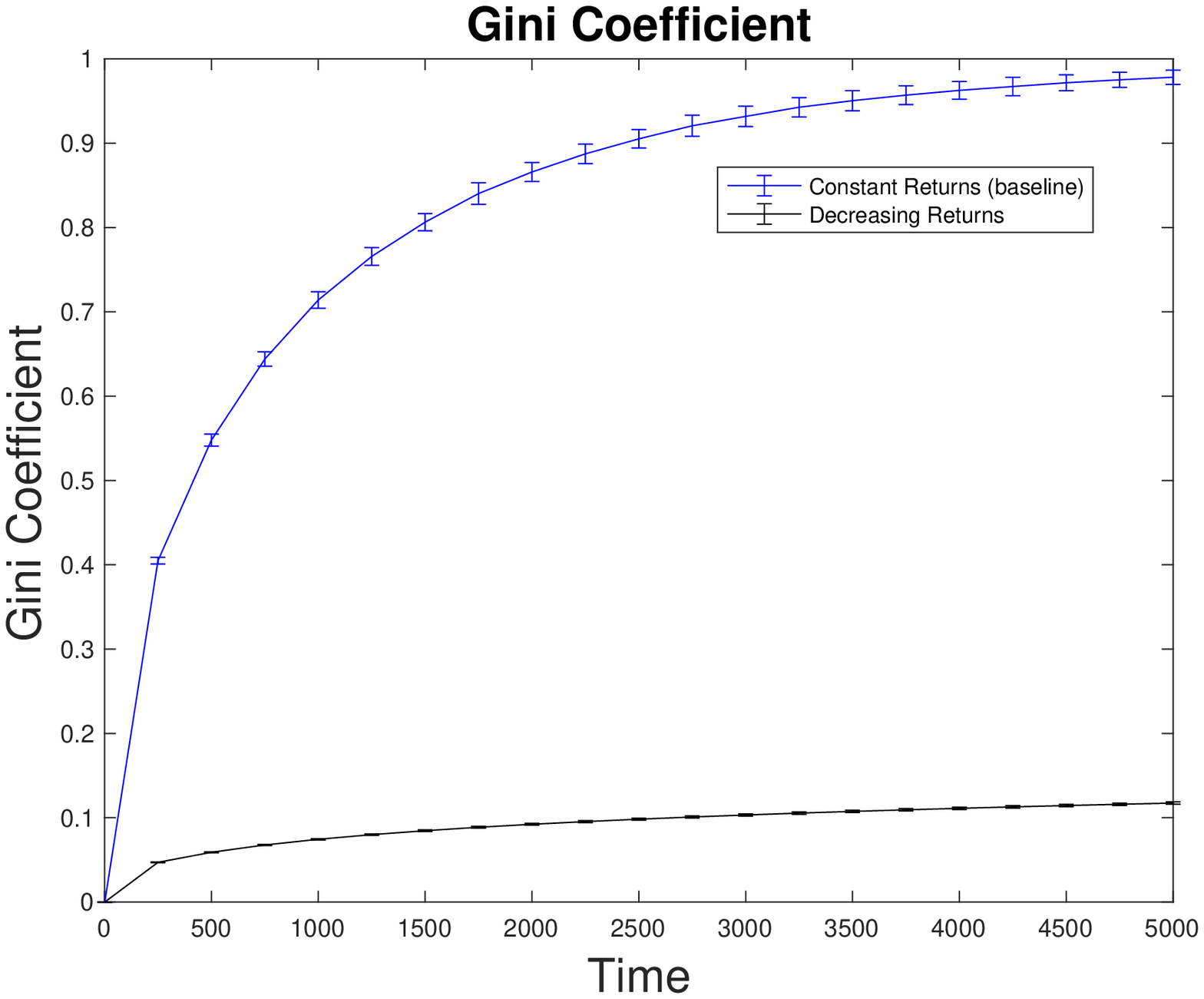}
\includegraphics[width=0.48\textwidth]{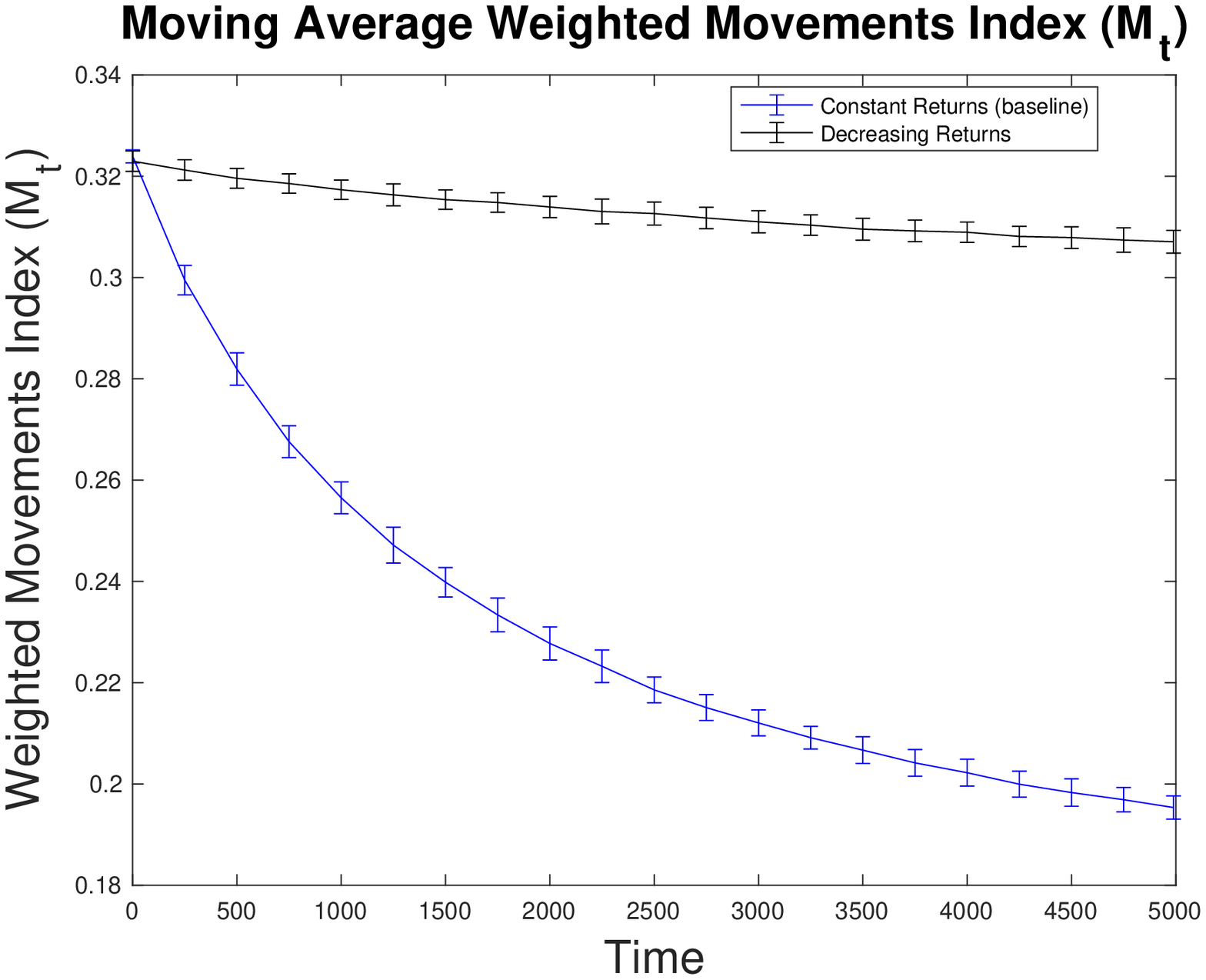}
\caption{Left Panel: Comparison of Gini Index over time under constant and decreasing returns to aggregate wealth. Right Panel: Comparison of Weighted Mobility Index $M_t$ over time under constant and decreasing returns to aggregate wealth. Averages and standard deviations are computed out of 100 simulations. In both panels and configurations: the simulations are computed under the baseline return structure with $r_{i,t} \sim N(0.05,0.05)$, and initial wealth $W_{i,t=0}=10 \,\,\, \forall i$. Under decreasing returns $r_{i,t}$ are modified according to Equation \ref{reduced}. } 
\label{Fig3B}
\end{figure}

This preliminary conclusion is corroborated by analysing simple return structure, which further allows to investigate the distinctive impact of the cumulative dimension of the baseline economic process. Simple return constitutes a reference logic which is diametrically opposite to compound return. It heuristically corresponds to the `financial capital maintenance' rule that is applied by corporate accounting systems. This rule computes the maintenance of invested shareholder equity from one period to the next one. The net difference is then reported as net earnings that can be distributed to shareholders (if positive), the undistributed part being accounted for as retained earnings. At the microeconomic level, simple return means that only net wealth is reinvested over time, while eventual net gains are consumed period after period. This heuristically corresponds to investing the same nominal amount in bonds and gilts repeatedly, without reinvesting proceeds through time (\citealt{biondi2011cost}). At the macroeconomic level, simple return may correspond to a capital formation process that is stationary and constrained over reinvestment. Capital stock is then reproduced rather than accumulated, its net contribution to total income being consummated over time. Under simple return structure, financial capital is  {\it remunerated} as productive factor but it is not {\it financially accumulated} over time. It is then treated as a flow factor (becoming analogous to labor in this respect), involving full consumption of all the net positive proceeds generated by capital wealth at each period of time. Formally: 
\begin{revs}
\begin{equation}
W_{i,T}=W_{i,1}  (1+\sum_{t=1}^T r_{i,t}) \,\,\,\, \text{with} \,\,\, r_{i,t}\geq-1 \,\,\, \forall \,\, i,t
\label{simplereturnstotal}
\end{equation}
\end{revs}

For sake of simulation, we compute actual simple returns under two return structures: $N(\mu_r, \sigma_r)$ and $gamma(a,b)$.
Computational results (Figure \ref{Fig4B}) show that wealth inequality does not accumulate under simple return structures, while wealth mobility remains virtually constant over time. Individuals go on adding their profits and losses to the initial invested capital, without reinvesting the proceeds. Asymptotically, those profits and losses compensate each other because they are generated by an additive stochastic process applied over the same initial amount. In particular, the Gini Index is decreasing and asymptotically near to zero, while the Wealth Mobility Index remains asymptotically higher than the baseline case. In Appendix \ref{tendtozero}, we analytically prove the following Lemma (as visualised by simulations in Figure \ref{Fig4B}):

\begin{lemma}
The Gini Index $G_t$ tends asymptotically to 0 for $t \rightarrow \infty$, under simple return structure, implying perfect equality among individuals.
\end{lemma}

\begin{figure}[!ht]
\centering
\includegraphics[width=0.48\textwidth]{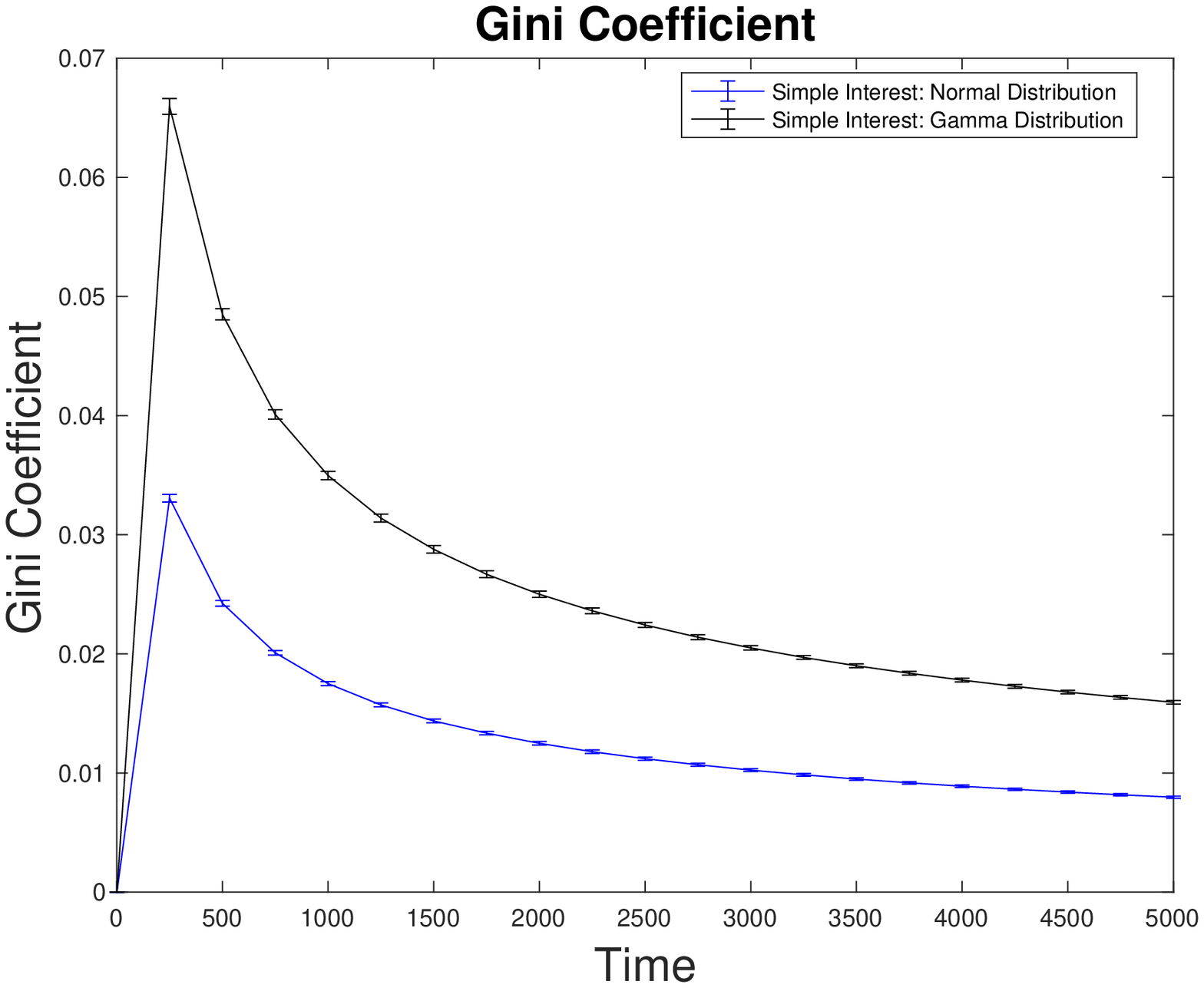}
\includegraphics[width=0.48\textwidth]{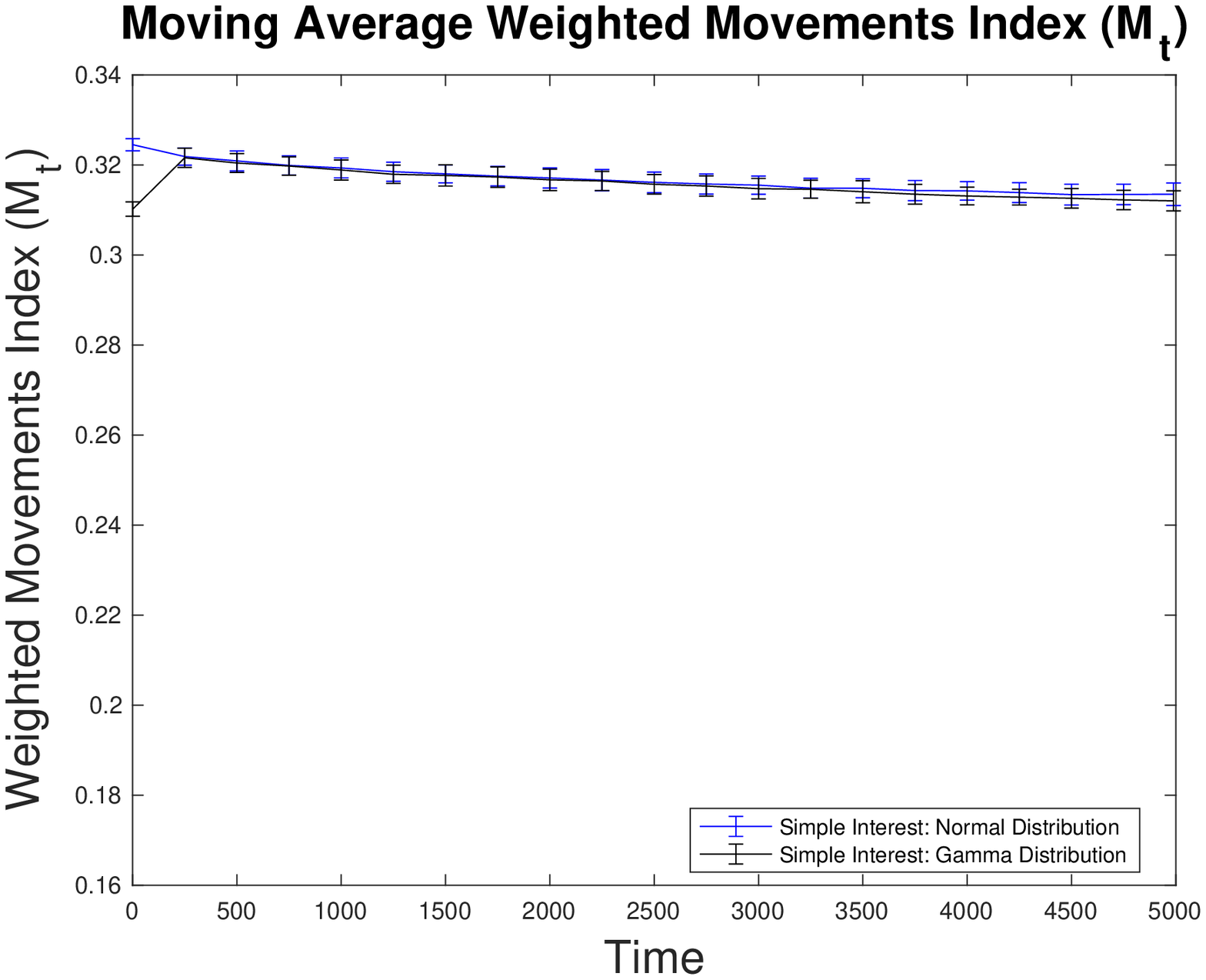}
\caption{Left Panel: Gini Index over time under simple return structure with normal distribution $N(\mu_r, \sigma_r)=N(0.05;0.05)$, and under simple return structure with Gamma distribution $gamma(a,b)=gamma(0.25;0.2)$.  The Gini Index asymptotically tends to zero (see Appendix \ref{tendtozero}). Right Panel: Moving average of Weighted Mobility Index under simple return structure with normal distribution $N(\mu_r, \sigma_r)=N(0.05;0.05)$, and simple return structure with Gamma distribution  $gamma(a,b)=gamma(0.25;0.2)$. Averages and standard deviations are computed out of 100 simulations. In all cases, initial wealth $W_{i,t=0}=10 \forall i.$}
\label{Fig4B} 
\end{figure}

Therefore, wealth inequality proves to be crucially dependent on the accumulation of returns through time, while this accumulation further undermines social mobility as expressed by relative levels of wealth.
In sum, the economic process modelled by Levy et al. denotes a significant connection between inequality and the financial accumulation process. This accumulation through time proves to be conducive to increased wealth inequality and decreased wealth mobility over historical time. However, its assumption of constant average returns to aggregate wealth under compound return structure seems unsustainable, because indefinite compounding cannot realistically hold in the long-run  (\citealt{voinov2007reconciling,biondi2011cost}). Moreover, IMF studies (\citealt{berg2013inequality,ostry2014redistribution}) show that inequality affects growth, with higher inequality being associated with lower growth.  It seems then unrealistic to assume a structurally stable economic process while the upper wealth tail goes on appropriating an ever-increasing part of total wealth, with the Gini Index asymptotically reaching the maximum value of one in the long-run. However, our model  clearly shows the logical and institutional tensions between the multiplicative logic embedded in compound return, and this potential trade-off between inequality and growth. Let alone, such a multiplicative process may involve a self-fulfilling decrease in returns, reducing then social welfare.

In order to extend this one-factor model, the next section shall introduce a two-factors model of the economic process, adding a flow factor which features an additive evolution to the stock factor characterised by a cumulative evolution through time.

\section{A model of aggregate economic process combining capital wealth (stock) and labour (flow) factors}\label{sec5}

According to \cite{oulton1976inheritance}, policy attitudes towards the inequality of wealth depend on views of its two paradigmatic causes. Accordingly, wealth inequality can be explained either ``because income was unequally distributed and hence some people saved more, in consequence accumulating more wealth'', or because wealth inheritance (accumulation) through time. Textbook macroeconomics introduces a stylised economic process that combines both economic factors: one stock factor (capital wealth) and another flow factor (labour). In this context, the Cobb-Douglas is a classic function of production for the aggregate economy, featuring factor returns to scale in its parameter space.
Extending the baseline model by Levy et al., a second flow factor can be introduced as follows. Total income comprises the sum of wealth income and labour income as follows:
\begin{equation}
Y^T_{i,T} = Y^W_{i,T} + Y^L_{i,T}
\end{equation}

Labour income can be consumed (for a share $1- s_{i,t}$) or saved (for a share $s_{i,t}$) at every period $t$. Total wealth comprises then accumulated wealth at compound returns and cumulated saved income as follows:
\begin{equation}
W^{SUM}_{i,T} =W_{i,t=1}\prod_{t=1}^T (1+ r_{i,t}) + \sum_{i,T} (s_{i,t} \cdot Y^L_{i,t})
\label{Wsum13}
\end{equation}
with $r_{i,t} \geq -1$ denoting pure wealth yield and $0 \leq s_{i,t} <1$ denoting saving share.

In this context, total wealth growth can be defined as follows:
\begin{equation}
 \text{Growth}_t = \frac{TW_t - TW_{t-1}}{TW_{t-1}}
 \label{Growth}
\end{equation}
where $TW_t$ applies Equation \ref{tw10} to $W^{SUM}_{i,T}$ defined in Equation \ref{Wsum13}.
For sake of generality and comparability, we introduce a fair condition between capital wealth yields and the labour income savings by setting the initial capital wealth so that it yields a weighted-mean permanent rent equal to the weighted-mean saved income at each period of time $t$. For this purpose, the perpetual rent value (to be equaled with the initial capital wealth) is computed by the ratio between the period rent income and its period return rate.
 
According to our configuration, the wealth return  $r_{i,t}$ defines the wealth yield $Y^W_{i,t} = r_{i,t}W_{i,t-1}$. The fair condition imposes that this wealth income is equal, on average, to the saved income of the period. The latter is defined as $Y^L_{i,t} = s_{i,t} Y^L_{t}$.

Accordingly, the fair condition imposes that:
\begin{equation}
W_{t=1} \equiv \frac{\sum_{i=1}^N (s_{i,t}\cdot Y^L_{i,t})}{\sum_{i=1}^N (r_{i,t})}
\end{equation}

Under this condition, our parameter space does not introduce any bias between the respective period incomes generated by the two factors on average (Figure \ref{Fig9}). At the onset, both factors contribute, on average, equally to total income.

\begin{figure}[!ht]
\centering
\includegraphics[width=0.48\textwidth]{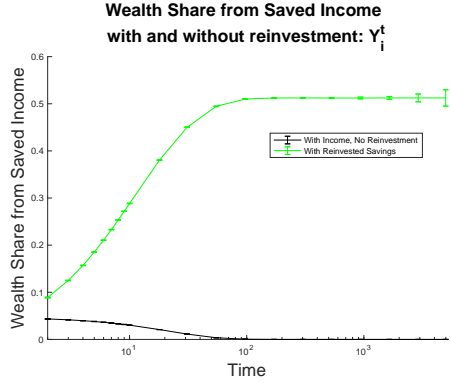}
\caption{Time evolution of the share of wealth generated by labour income and labour-income savings. Averages and standard deviations are computed out of 100 simulations.  See Section \ref{sec5} for details about the simulation conditions.}
\label{Fig9} 
\end{figure}
On this basis, we introduce labour income and savings with the same logic which Levy et al. apply to capital wealth. We introduce a stochastic saving rate  $s_{i,t} \in [0,1]$ for which:
\begin{itemize}
\item[] (i) all saving rates are equally likely, and
\item[] (ii) the average saving rate is always equal for all the individuals.
\end{itemize}
Under this assumption, taking $Y^L_{i,t}= 1$ and $s_{i,t} \sim U[0,1]  \,\, \forall i,t$ and recalling that: $\frac{1}{N}\sum_{i=1}^N (s_{i,t} Y^L_{i,t})=s_{mean} \equiv 0.5$ while $\frac{1}{N}\sum_{i=1}^N (r_{i,t})=\mu_r\equiv 0.05$, the fair condition imposes $W_{t=1} \equiv \frac{0.5}{0.05}=10$ consistently with our assumption for initial wealth $W_{i,t=1} = 10$ for all individuals.

Our computational analysis shows that the presence of a labor factor does not reshape the baseline economic process concerning wealth inequality and mobility. In particular, either if labour income cannot be saved  (that is, $s_{i,t} = 0 \,\,\forall  i,t$), or if savings cannot be reinvested (as in the previous Equation \ref{Wsum13}), the additive process of labour factor cannot match the multiplicative process of wealth accumulation. Consequently, the latter continues to dominate the aggregate economic process in the long run. Wealth inequality and mobility are not reshaped by the presence of that additive factor (Figure \ref{Fig5B}). 

\begin{figure}[!ht]
\centering
\includegraphics[width=0.48\textwidth]{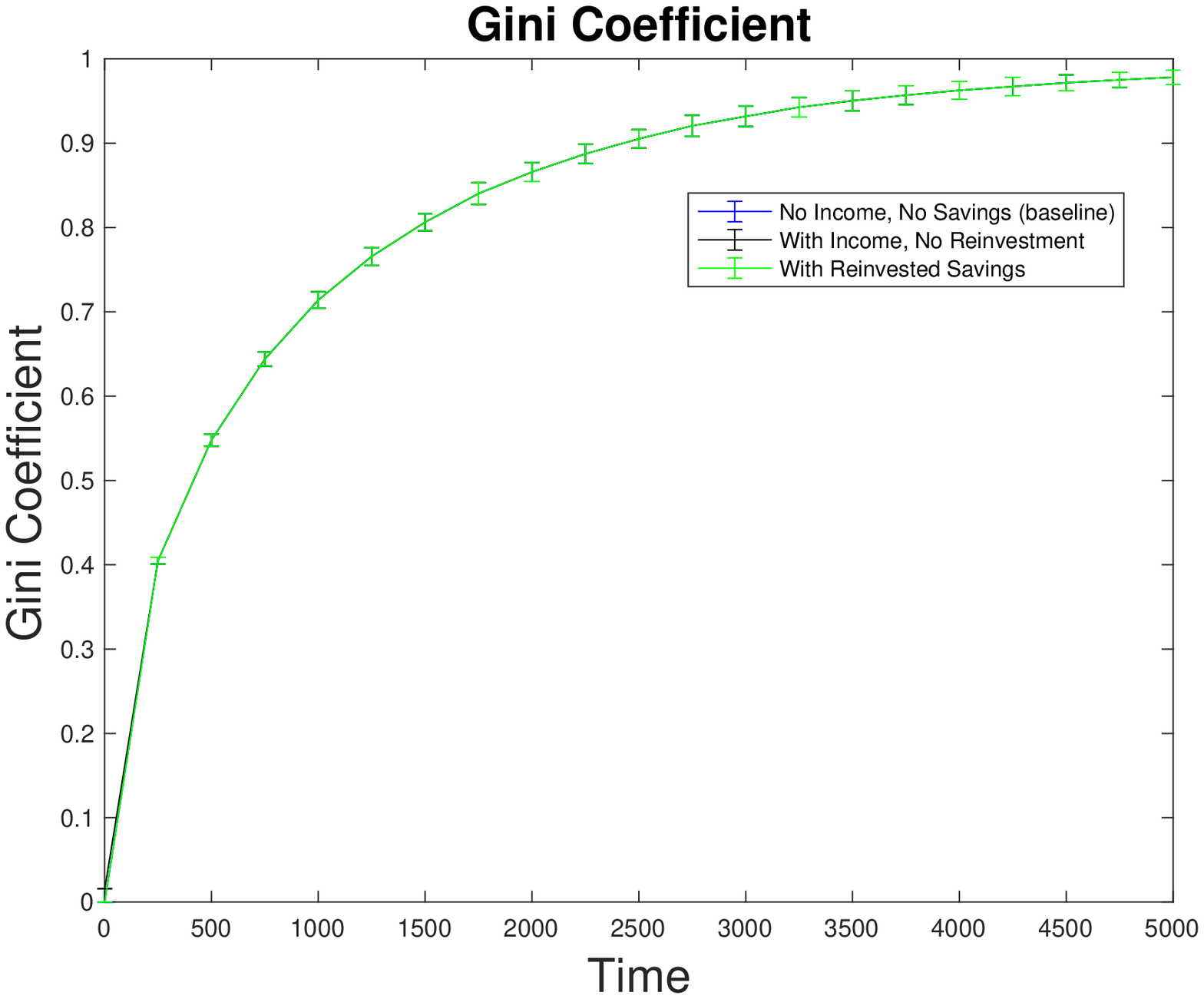}
\includegraphics[width=0.48\textwidth]{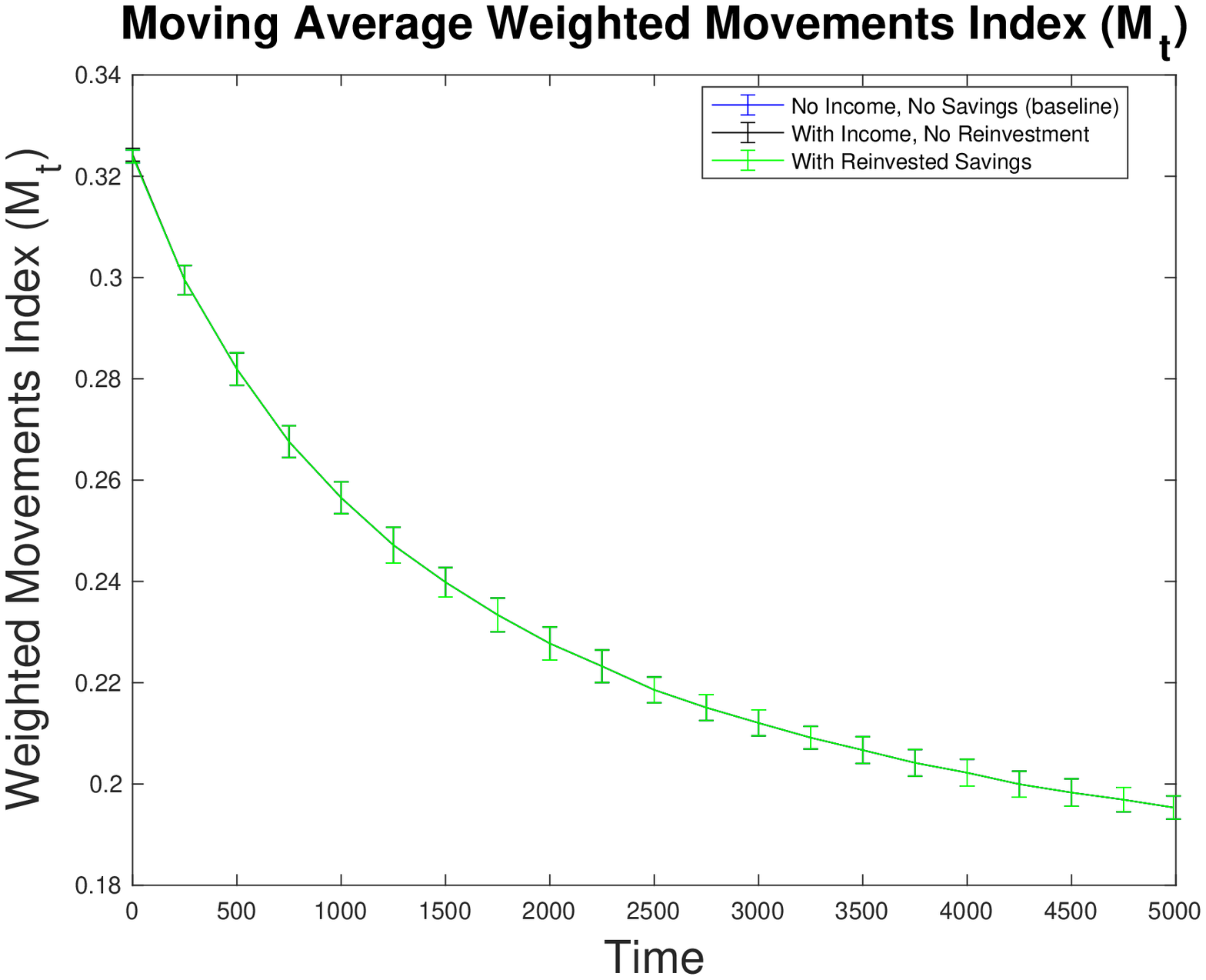}
\caption{Left Panel: Gini Index value over time under baseline case (one-factor model), twofactors model without labor-income savings, and two-factors model with reinvested savings over time. Right Panel: Moving average of Weighted Mobility Index under baseline case (one-factor model), two-factors model without labor income savings, and two-factors model with reinvested savings over time. Averages and standard deviations are computed out of 100 simulations. In all cases, initial wealth  $W_{i,t=0}=10 \forall i.$}
\label{Fig5B} 
\end{figure}

Labour income alone cannot reshape wealth concentration, inequality and mobility dynamics through time. Nevertheless, some may wonder whether reinvested savings from labour income could. To test this hypothesis we introduce a progressive wealth accumulation driven by labour income savings. Formally:
\begin{equation}
W_{i,t+1} = W_{i,t}(1+r_{i,t}) + (s_{i,t+1} Y^L_{i,t+1})
\end{equation}
In this case, the saved share of labour income $Y^L_i$ is progressively accumulated through compound returns along with inherited wealth across individuals and over time. Total individual wealth through time includes both saved wealth and inherited wealth, formally:
\begin{equation}
W^{SUM}_{i,t+1}=W_1 \prod_{h=1}^t (1+r_{i,h}) + \sum_{j=1}^{t+1} s_{i,j} Y^L_{i,j} \prod_{k=j}^t (1+r_{i,k})
\end{equation}
If labor income remains constant over time, i.e. $Y^L_{i,t}=Y_0 \,\,  \forall i,t$, this formula reduces to:
\begin{equation}
W^{SUM}_{i,t+1}=W_1 \prod_{h=1}^t (1+r_{i,h}) + Y_0\sum_{j=1}^{t+1} s_{i,j}  \prod_{k=j}^t (1+r_{i,k})
\label{totalwtwofactorreduced}
\end{equation}
In Equation \ref{totalwtwofactorreduced}, we maintain that $r_{i,h}$ (returns over inherited wealth) have the same structure as $r_{i,k}$ (returns over accumulated savings). For sake of simulation, we continue to assume $Y_0 = 1 \,\, \forall i,t$.

According to our computational results (Figure \ref{Fig5B}), even this progressive reinvestment of savings cannot reshape wealth inequality and mobility under the fair condition stated above. Wealth distributions and all the indexes maintain the same behaviour as under the baseline case. \footnote{The same results hold when actual returns are derived from simple return structure. Computational results are available under request.}

Although each individual draws savings from labour income according to a uniform distribution, these savings are reinvested according to the same financial accumulation process as the inherited wealth, which is initially equal for all agents. Therefore, savings from labour income do not introduce an alternative reference logic or a complementary economic process. Consequently, they cannot reshape wealth distribution, inequality and mobility over time.

In sum, under both the Levy et al. model and the widespread two-factors model of the aggregate economic process, wealth concentration, inequality and mobility depend crucially on the compound return structure that characterises the accumulation of financial investment over time. The decomposition of wealth dynamics in two factors of productions does not change its evolution over time.

In fact, our computational results further show emerging aggregate behaviour that proves distant from economic reality, echoing the \citealt[(p. 81)]{kaldor1938theory}'s claim that:
\begin{itemize} \item[] The entire notion of `factor of production' is an incubus on economic analysis, and should be eliminated from economic discussion as summarily as possible. \end{itemize}
In a similar vein, \cite[(p. 138)]{solow1976} acknowledged the aggregation problem as follows:
\begin{itemize} \item[] I have to insist again that anyone who reads my 1955 article [\cite{solow1955production}] will see that I invoke the formal conditions for rigorous aggregation not in the hope that they would be applicable [. . .] but rather to suggest the hopelessness of any formal justification of an aggregate production function in capital and labor.  
\end{itemize}

Nevertheless, our computational analysis shed some light on theoretical and applied implications that are implicitly assumed in widespread representations and models of the aggregate economic process, as reviewed and further developed by \cite{fernholz2014instability} and \cite{bertola2006income} among others. From this perspective, the following section shall expand upon our preliminary conclusions about inequality, mobility and the financial accumulation process. We shall introduce a stylised institutional configuration that typically frames and shapes income and wealth dynamics in economy and society: taxation.

\section{The impact of taxation and redistribution}\label{sec6}

Taxation is a classic matter related to income and wealth distributions. Available statistics do extensively rely upon fiscal data for gathering evidence (\citealt{SaezZucmanslide,topritzhofer1970outline}). It is then interesting to explore its effect on wealth concentration, inequality and mobility by expanding the baseline model by Levy et al. through featured modes of taxation.

For simulation purpose, we introduce four stylised modes of taxation, all governed by a central authority that knows and intervenes over the individual positions of wealth and income (change in wealth) at the end of each period. Four modes of taxation combine two methods of tax levy with two methods of tax distribution:
\begin{itemize}
\item Concerning tax levy, we assume either a uniform proportional tax rate for all individuals (uniform taxation), or a progressive tax rate increasing with the tax base (progressive taxation).
\item Concerning tax distribution, we assume either a uniform redistribution to all individuals (featuring the provision of universal public service), or a regressive redistribution that decreases with the tax base (featuring provision of direct transfers for welfare to the polity).
\end{itemize}
Under proportional taxation and public service model (\textbf{Proportional \& PS}), the tax authority levies a fixed universal share $\tau_{i,t} = \tau \,\,\forall \,\, i,t$ on positive net changes in wealth $\mathcal{T}_{i,t}=\max\left\{W_{i,t}-W_{i,t-1}; 0\right\}$, and redistributes the total levied amount equally among all the individuals. This model features the provision of public services to the polity through proportional taxation as follows:
\begin{equation}
Tax_{i,t}=\mathcal{T}_{i,t}\tau;
\label{taxproportional}
\end{equation}

\begin{equation}
Subsidy_{i,t}=\frac{\sum_k Tax_{k,t}}{N} 
\label{redistributionLiberal}
\end{equation}

Under proportional taxation and welfare model (\textbf{Proportional \& Welfare}), the tax authority employs the proportional tax levy to provide direct transfers to the polity. These transfers are redistributed in regressive proportion to individual wealth. Individual tax is computed according to Equation \ref{taxproportional}, while the redistribution is managed according to the following formula:

\begin{equation}
Subsidy_{i,t}=\left[1-  \frac{\mathcal{T}_{i,t}}{\sum_{k=1}^N  \mathcal{T}_{k,t}} \right] \frac{1}{N-1} \sum_k Tax_{k,t}
\label{redistributionwelfare}
\end{equation}

Under progressive taxation and public service model (\textbf{Progressive \& PS}), the tax authority levies a progressive share $\mathcal{T}_{i,t}$ of net changes in wealth. Accordingly, the richest individual applies the maximum rate $\tau_{max}=\tau_{i,t}(\max \mathcal{T}) \,\, \forall t$, while the poorest individual does not pay anything. On this basis, the tax authority provides a universal public service to all the individuals. Individual tax payment is defined as follows:
\begin{equation}
Tax_{i,t}=\mathcal{T}_{i,t} \cdot \tau_{max} \frac{\mathcal{T}_{i,t} - \min \mathcal{T}_{i,t}}{\max \mathcal{T}_{i,t} - \min \mathcal{T}_{i,t}}
\label{taxprogressive}
\end{equation}
On this basis, the individual $i$ having the maximal tax base $\mathcal{T}_{i,t}$ applies the maximal tax rate $\tau_{max}$, while the individual $i$ having the minimal tax base $\mathcal{T}_{i,t}$ does not pay taxes (i.e., its tax rate $\tau_{i,t}=0$ at time period $t$). The redistribution mechanism under this model follows Equation \ref{redistributionLiberal}.

Finally, under progressive taxation and welfare model (\textbf{Progressive \& Welfare}), the tax authority employs the tax levy denoted by Equation \ref{taxprogressive} in order to redistribute the total levied amount in a regressive proportion to wealth, according to Equation \ref{redistributionwelfare}. This model features the provision of direct transfers to the polity, funded by progressive tax levy.

In sum, our computation analysis features four stylised modes of taxation, summarized in Table \ref{Table1}.

\begin{table}[h!]
\caption{Stylised modes of taxation and redistribution}
\centering
\footnotesize
\begin{tabular}{ | c | c | c |}
\hline
\multicolumn{1}{|r|}{\textbf{Tax redistribution}}  & \textbf{Proportional} & \textbf{Regressive} \\
\multicolumn{1}{|l|}{\textbf{Tax Levy}} & \textbf{Redistribution Rate} & \textbf{Redistribution Rate} \\ \hline
\textbf{Proportional Tax rate} & Proportional taxation $\with$ public service & Proportional taxation $\with$ welfare \\ \hline
\textbf{Progressive Tax Rate} & Progressive taxation $\with$ public service & Progressive taxation $\with$ welfare \\ \hline
\end{tabular}
\label{Table1}
\end{table}

We assess the impact of taxation over wealth inequality and mobility under these stylised modes of taxation. 
For sake of simulation, we retain a tax rate $\tau=0.05$ for the proportional tax system and a maximum tax rate of $\tau_{max}=0.10$ for the progressive taxation. Under progressive tax regimes, this framework implies a tax-rates structure that endogenously depends on the distributions of income (net wealth change) and wealth and their evolution over time. Computational results (Figure \ref{Fig6B}, \ref{Fig6D}) show that all modes of taxation and redistribution are effective in reducing and stabilising wealth inequality (Figure \ref{Fig6B}, Left Panel), while improving and stabilising wealth mobility (Figure \ref{Fig6B}, Right Panel). Contrary to savings from labour income, taxation effectively introduces an alternative economic logic and a complementary economic process in our miniature economy. This alternative and complementary collective action proves to be effective in compensating the impact of financial accumulation over wealth inequality and mobility.

\begin{figure}[!ht]
\centering
\includegraphics[width=0.48\textwidth]{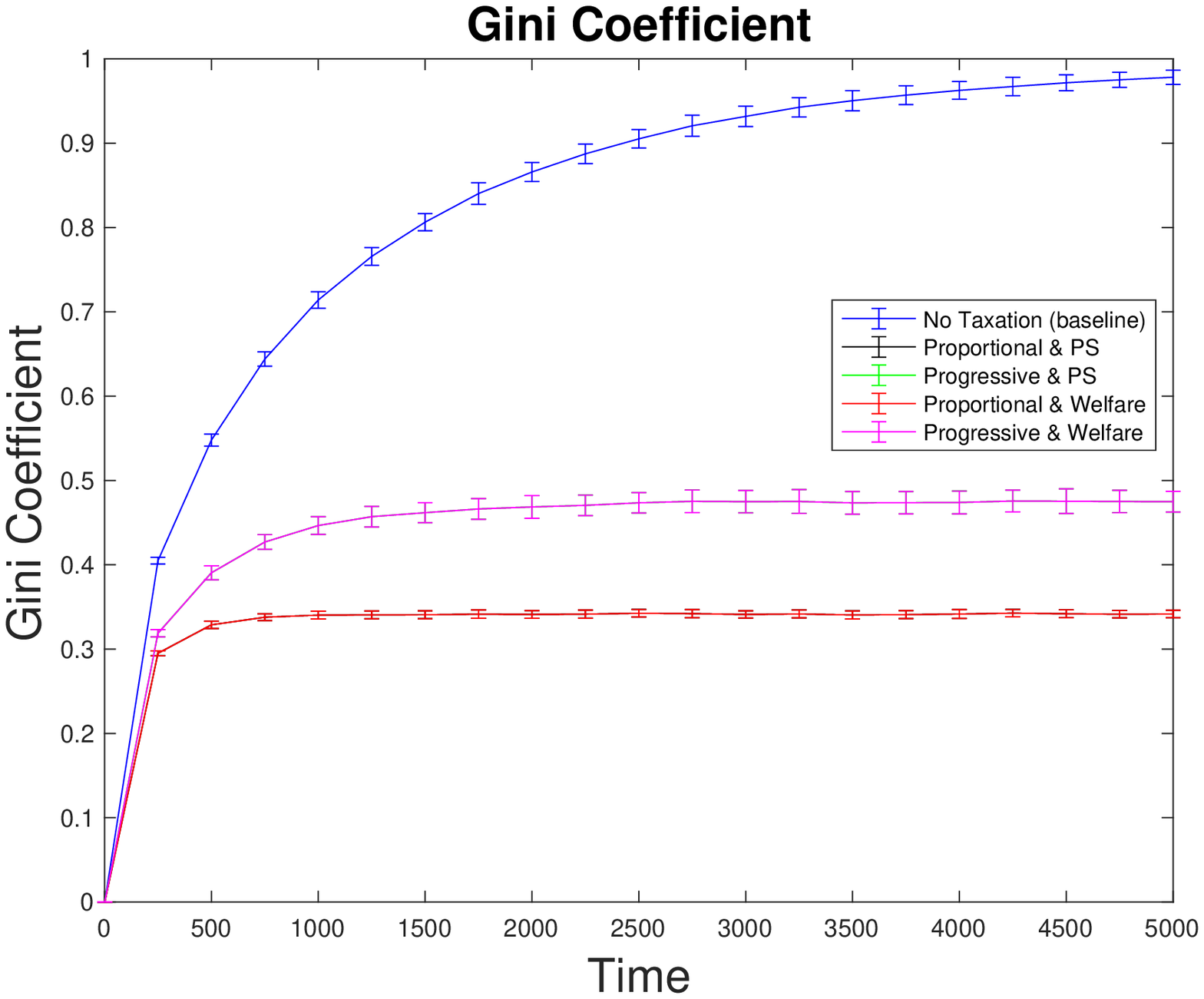}
\includegraphics[width=0.48\textwidth]{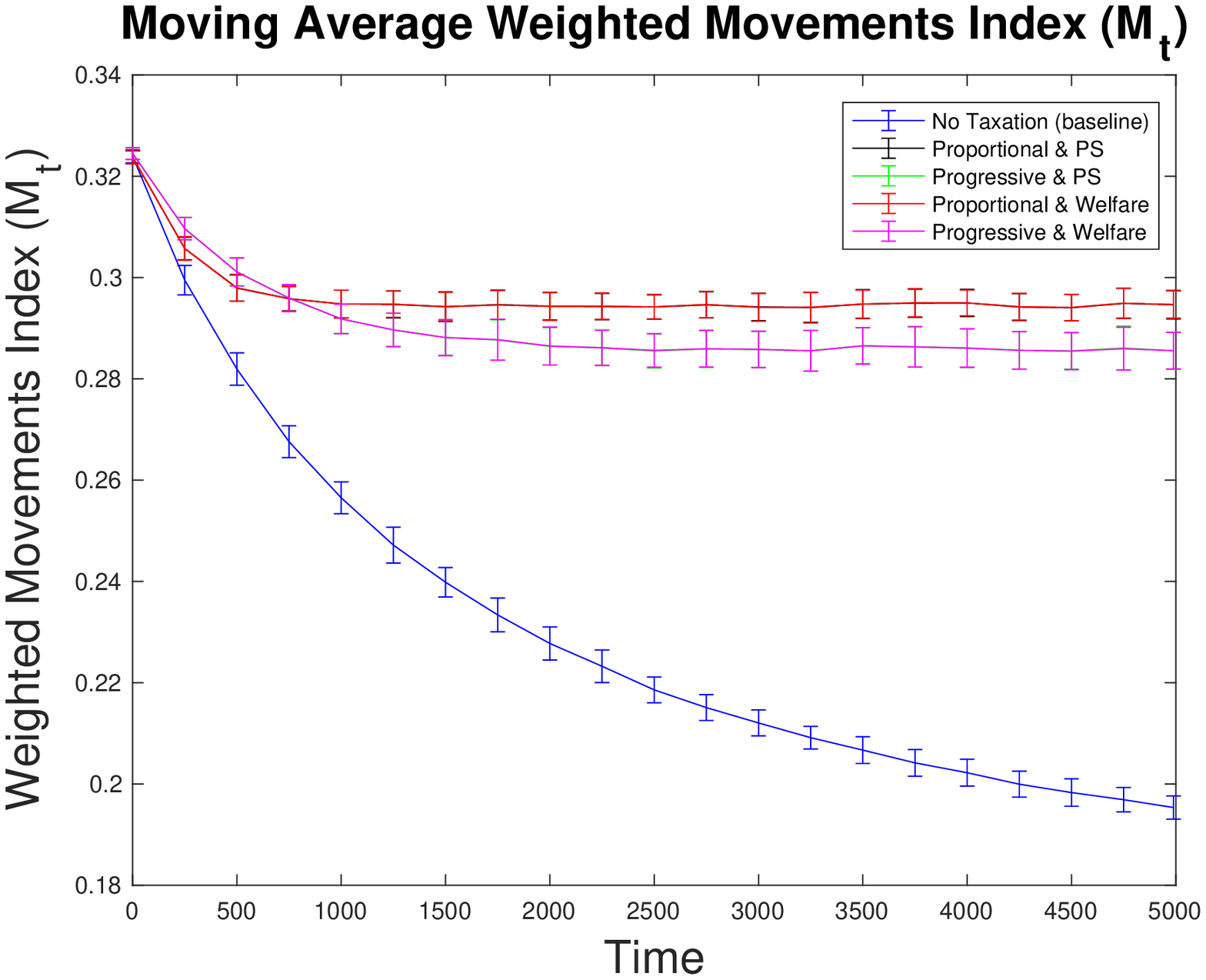}
\caption{Left Panel: Gini Index over time under baseline case (no taxation), proportional taxation and public service model (\textbf{Proportional \& PS}), proportional taxation and welfare model (\textbf{Proportional \& Welfare}), progressive taxation and public service model (\textbf{Progressive \& PS}) and progressive taxation and welfare model (\textbf{Progressive \& Welfare}). Right Panel: Weighted Mobility Index over time under the same cases. Averages and standard deviations are computed out of 100 simulations.  In all cases, initial wealth $W_{i,t=0}=10 \,\,\, \forall i$. See Section \ref{sec6} for details about implementation of  taxation and redistribution systems.}
\label{Fig6B} 
\end{figure}

In particular, Gini Index is consistently and materially inferior to the baseline case, while it remains asymptotically far from one. Wealth mobility indicators are consistently and materially superior to the baseline case featured by the Levy et al. model.

\begin{figure}[!ht]
\centering
\includegraphics[width=0.48\textwidth]{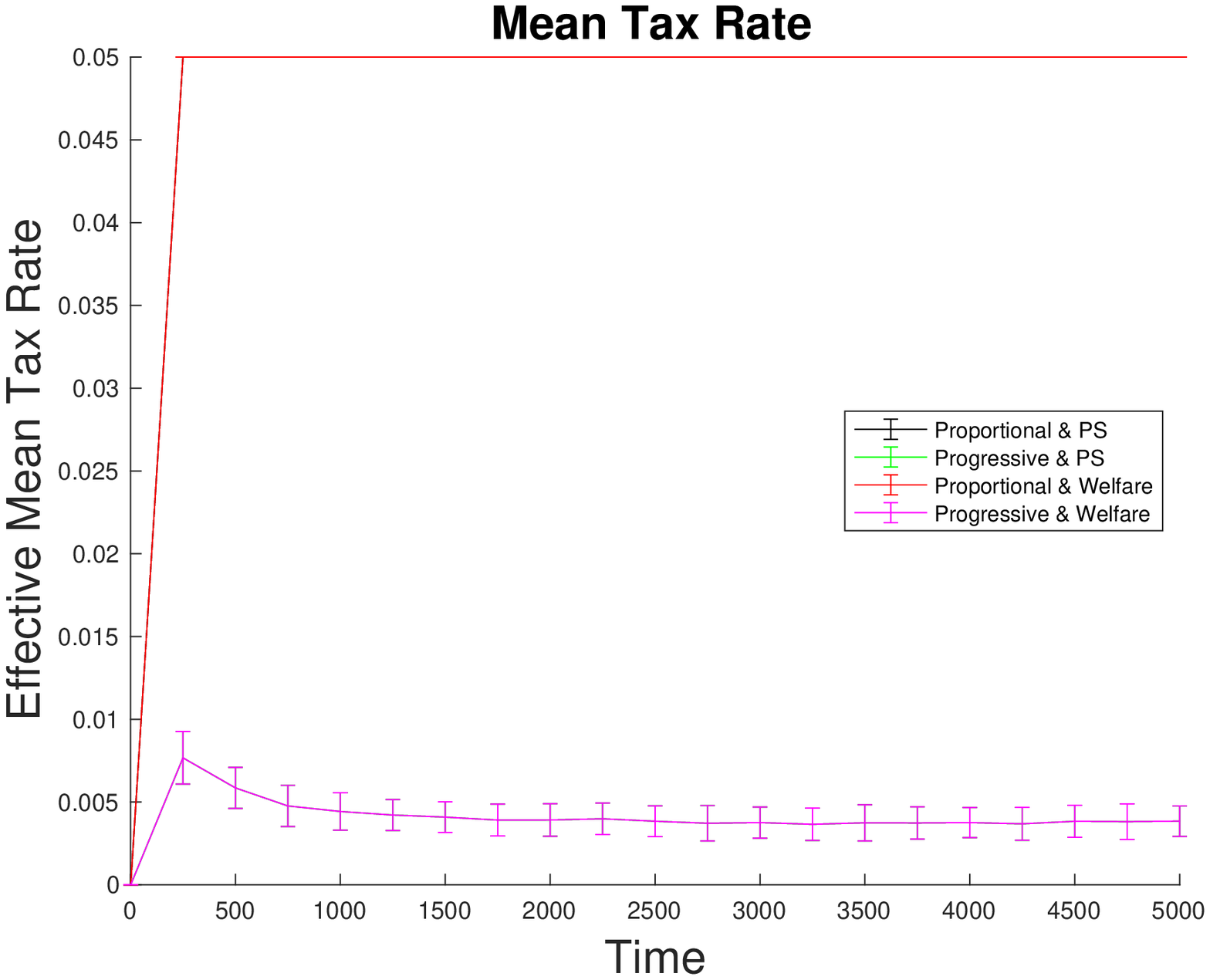}
\includegraphics[width=0.48\textwidth]{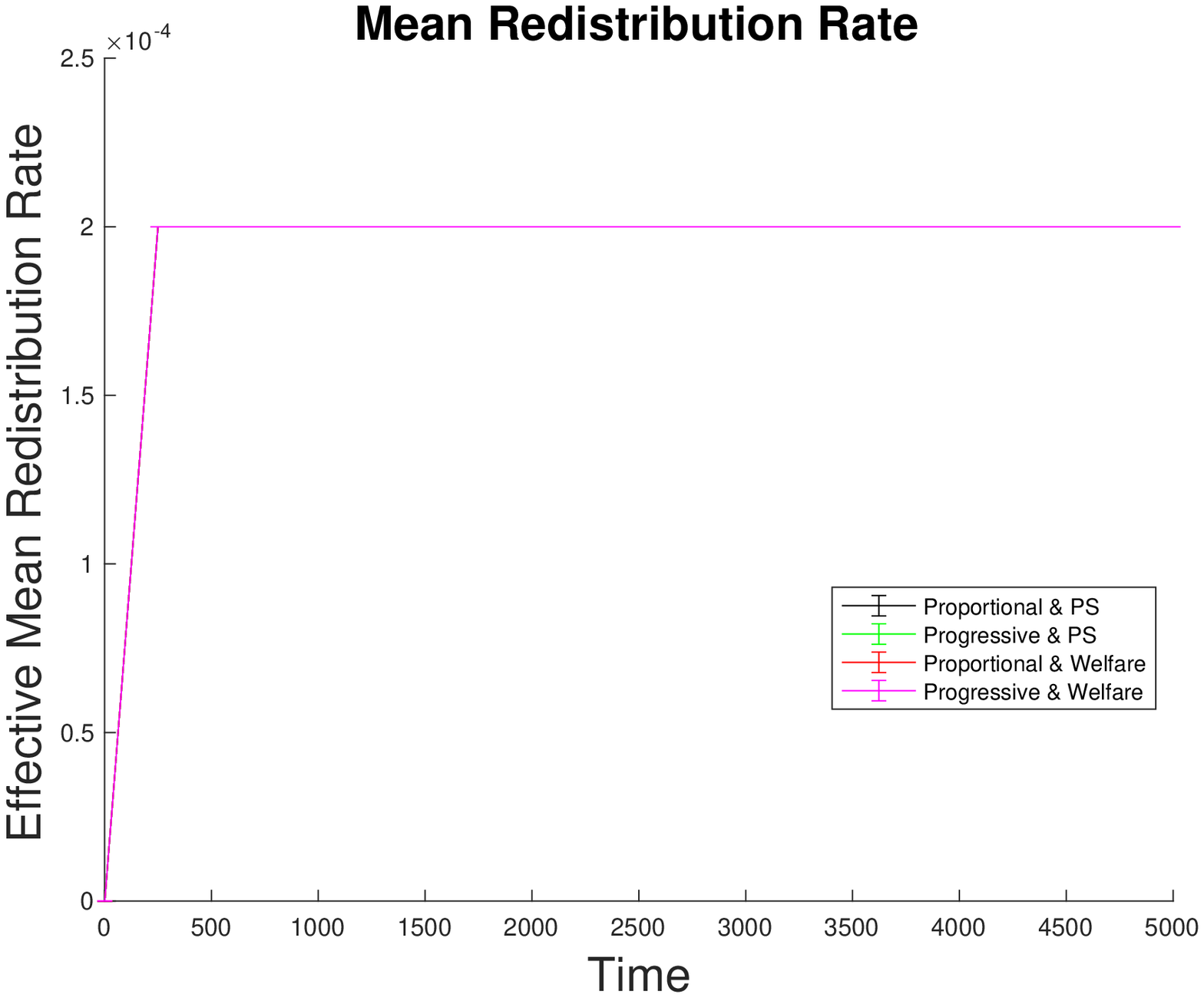}\\
\includegraphics[width=0.48\textwidth]{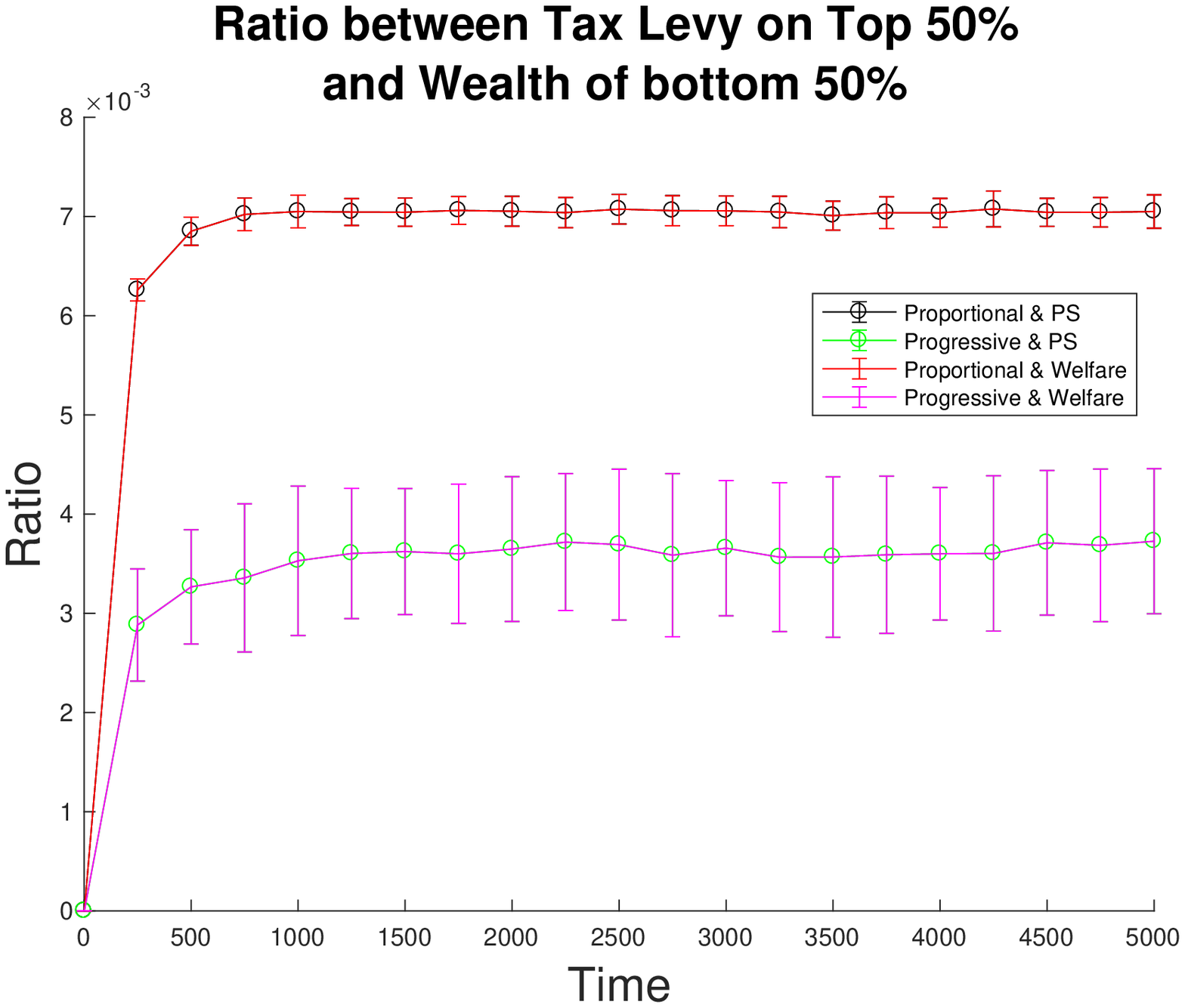}
\caption{Left Panel: Mean Tax Rate over time under baseline case (no taxation), proportional taxation and public service model (\textbf{Proportional \& PS}), proportional taxation and welfare model (\textbf{Proportional \& Welfare}), progressive taxation and public service model (\textbf{Progressive \& PS}) and progressive taxation and welfare model (\textbf{Progressive \& Welfare}). Right Panel: Mean Redistribution rate over time under the same cases. Bottom Panel: Ratio between the total tax levied on the richest 50\% of the population, relative to the total wealth of the bottom 50\% of the population. Averages and standard deviations are computed out of 100 simulations. In all cases, initial wealth $W_{i,t=0}=10 \,\,\, \forall i$. See Section \ref{sec6} for details about implementation of  taxation and redistribution systems.}
\label{Fig6D}
\end{figure}

The actual relative impact on wealth inequality and mobility across modes of taxation depends on the parameter space assumptions. Computational results (Figure \ref{Fig6D}) for mean tax and redistribution rates explain why progressive taxation is less effective than proportional taxation in reshaping wealth inequality and mobility, under our framework of analysis. As mentioned above, individual tax rates under progressive tax regimes depend on the underlying distribution of wealth. Since the latter is increasingly right-skewed over time, this dependency involves a mean tax rate that progressively becomes and remains very low over time (materially inferior to the mean tax rate of 0.05 applied under proportional tax regimes). Consequently, the impact of progressive tax regimes over wealth inequality is materially reduced both in absolute terms, and relative to proportional tax regimes which apply an exogenous fixed tax rate. Moreover, since wealth distribution is increasingly and materially right-skewed over time, a relatively low tax rate is sufficient to asymptotically stabilise the Gini Index (Figure \ref{Fig6B}, Left Panel). Wealth is so concentrated on the top (see Lemma \ref{prop1}) that a relatively low tax extraction from the richer is sufficient to materially increase wealth of the poorer, involving a stabilising effect of wealth inequality over time (Figure \ref{Fig6D}, Bottom Panel). This result does not establish preference for, or superiority of proportional tax regimes over progressive tax regimes. Indeed, the effectiveness of the proportional tax system, while guaranteeing higher tax collection, neglects the regressive nature of this tax mechanism.  A policy implication of this result is that effectiveness of fiscal systems depends on the underlying economic structure and process. Therefore, our analysis would recommend tax authorities committed to progressive tax regimes to maintain tax rate structures based on absolute wealth thresholds and independent from relative wealth levels. The latter tax authorities should secure a sufficient degree of progressiveness of taxation, as well as a sufficiently high top tax rates.

In conclusion, taxation materially reduces wealth concentration and inequality, compensating the impact of financial accumulation process. Taxation proves therefore to be effective in counterbalancing the inequality effects of the financial accumulation process.  This result is consistent with \cite{fernholz2014instability} arguing that ``the presence of redistributive mechanisms then ensures the stability of the distribution of wealth over time''.  

\section{Concluding remarks}
The poet Trilussa mocked national statistics to be that accounting method for which, one individual having eaten two chickens and another one just none, both would result to have eaten one chicken each.\footnote{``Me spiego: da li conti che se fanno $\backslash$ seconno le statistiche d'adesso $\backslash$
risurta che te tocca un pollo all'anno: $\backslash$ e, se nun entra nelle spese tue, $\backslash$
t'entra ne la statistica lo stesso $\backslash$ perch\`{e} c'\`{e} un antro che ne magna due."  (Trilussa, La statistica)}
 Students of income and wealth distributions may keep this adage in mind while developing related macroeconomic modelling, especially under the representative agent assumption.
 
Our computational economic analysis shows the significant connection between inequality and the financial accumulation process in the study of income and wealth distributions. This connection has been investigated through progressive extensions of the baseline model introduced by Levy et al.
Our analysis shows the limited heuristic contribution of a two factors model comprising one single stock (capital wealth) and one single flow factor (labour) as pure drivers of aggregate income and wealth generation and allocation over time. We further show the theoretical contribution of minimal institutions (\`{a} la Shubik), to partly overcome this limitation. In particular, we investigate heuristic models of taxation in line with the baseline approach. Drawing upon our computational economic analysis, we can infer that the financial accumulation process plays a significant role as socioeconomic source of inequality, while institutional configurations including taxation play another significant role in framing and shaping the aggregate economic process that evolves over socioeconomic space and time. Our computational economic analysis is based upon a simple modelling strategy combined with a calibration that is suitable for comparing alternative model configurations. This calibration does not necessarily fit empirical regularities. Therefore, we cannot infer empirical or forecasting predictions, but rather theory-driven implications that deserve further consideration from theoretical and applied viewpoints. Wealth inequality and mobility are important socio-economic dimensions of our economy and society. Increased wealth inequality may raise fairness issues, undermining economic sustainability and development through historical time. Decreased wealth mobility may raise further fairness issues, undermining socio-economic incentives to entrepreneurship and workmanship. Concerning wealth inequality and mobility issues, our computational economic analysis points to featuring drivers that deserve further attention by researchers and policy-makers. First of all, the financial accumulation process appears to be the key driver of both issues, generated by the peculiar compound return structure that characterises financial investment in widespread institutional configurations. Its contribution to wealth inequality and mobility further appears to fundamentally depend on the financial market dynamics featuring volatility clustering and extreme events. Labour income and savings do not appear to be able to rebalance the impact of this financial accumulation process through historical time. Contrastingly, taxation appears to be effective in compensating its effect. Finally, according to our computational economic analysis, the causes of recent increases in wealth inequality may be sought in socioeconomic transformations of financial market dynamics and taxation (including fiscal niches exploitation, tax avoidance and the flattening of tax progressiveness) over recent decades. From our theoretical perspective, return structure, volatility and exuberance in financial markets, as well as the working of fiscal systems are candidates to drive wealth inequality and wealth mobility in our economy and society.

\section*{acknowledgements}
We dedicate this article to the memory of Prof. Pierpaolo Giannoccolo, coauthor and dear friend of us. Pierpaolo had been a committed team member and was working with us on furthering the understanding of the financial economic process. We thank Alan Kirman, Thomas Piketty, Stefano Olla, Shyam Sunder and Marco Valente for their insightful comments and suggestions. Previous versions of this article were presented at the `International Economic Law and the Challenge of Global Inequality' Conference, the 20th Annual Workshop on the Economic Science with Heterogeneous Interacting Agents (WEHIA 2015), and at the Applied Economics Lunch Seminar, Paris School of Economics (2 June 2015).

\section*{Appendix}

\subsection{Proof of Lemma 1 - Gini evolution for Compound Interest Structure}
\label{tendtoone}

\begin{proof}
We aim at proving that, for compound interest structure, $G_t \rightarrow 1$ as $t \rightarrow +\infty$.

\vspace{0.3cm}

Our proof draws upon \cite{fernholz2014instability}.
Our model for compound return structure replicates the background structure of the \cite{fernholz2014instability} model, as represented by their Equation 10. Accordingly, with our notation:
\begin{equation}
W_{i,t} = W_{i,t=1} \cdot e^{r(t)}
\end{equation}
This equation denotes continuously compound return structure over time, with return function $r(t)$ depending on a standard Brownian motion.
In this context, \cite{fernholz2014instability}'s Theorem 2 proves that, if $\sigma_r > 0$, the time-averaged share of total wealth held by the wealthiest single household converges to one, almost surely (their Equation 13), although it is not the same household to maintain the leading position over time. Analytically:
\begin{equation*}
\lim_{t\rightarrow +\infty} \int_0^t \theta_{i,t} (t) dt =1 \,\,\, \text{a.s.} 
\end{equation*}
with $\theta_{i,t} = \frac{W_{i,t}}{\sum_t W_{i,t}}$ and where $\theta_{max,t} = \max_{i} [ \theta_{i,t} ]$.

Our model applies a discretely compound return structure as follows: 
\begin{equation*}
W_{i,t} = W_{i,t=1}  \prod (1 + r_{i,t})
\end{equation*}

Without loss of generality, this structure can be made continuous by making the time change infinitesimal as follows: 
\begin{eqnarray}
W_{i,t} \rightarrow W_{i,t-1} \cdot  e^{R(t)}  \,\,\text{with} \,\, dt \rightarrow 0 \\
R(t) = \ln (1 + r_{i,t})
\end{eqnarray}
Where the $R(t)$ function transforms our return $r_{i,t}$ from discrete to continuous time.
This formulation is analogous to \cite{fernholz2014instability} formula, since $r_{i,t} \sim N (\mu_r ; \sigma_r)$ by construction. Their proof applies then to it.
\end{proof}

\subsection{Proof of Lemma 3 -- Gini evolution for Simple Interest Structure}
\label{tendtozero}

\begin{proof}
We aim at showing that, for simple interest structure, $G_t \rightarrow 0$ as $t \rightarrow +\infty$.

The Gini Index tend to 0 if the wealth of all individuals tends to be equal for $t \rightarrow +\infty$.
By construction, $\forall t$
\begin{revs}
\begin{equation*}
W_{i,t} = W_{i,1} (1 + \sum_t r_{i,t})
\end{equation*}
\end{revs}
The equality condition imposes that $W_{i,t} = W_{j,t}$, thus:
\begin{revs}
\begin{equation*}
W_{i,1} (1 + \sum_t r_{i,t} ) = W_{j,1} (1 + \sum_t r_{j,t} )
\end{equation*}
\end{revs}
or 
\begin{revs}
\begin{equation}
W_{i,1} (1 + \sum_t r_{i,t} ) - W_{j,1} (1 + \sum_t r_{j,t} )=0
\label{eq1prova}
\end{equation}
\end{revs}

Given that $W_{i,1} = W_{j,1} = W_{1}$, for each $i, j$ by assumption (i.e., equal initial wealth for all agents), Equation \ref{eq1prova} becomes:
\begin{revs}
\begin{equation}
W_1 + W_1 \sum_t r_{i,t}  - W_1 - W_1 \sum_t r_{j,t} = 0
\end{equation}
\end{revs}
or simplifying:
\begin{equation*}
W_1 [ \sum_t ( r_ {i,t}  ) -  \sum_t ( r_{j,t}) ] = 0
\end{equation*}
since $r \sim N(\mu_r,\sigma_r)$ by construction and since $\sum r_{h,t} \rightarrow t \cdot \mu_r$ if $t \rightarrow +\infty \,\,\, \forall h = i,j$, then:
\begin{equation*}
W_1 \left[ \sum_t (r_ {i,t}) - \sum_t (r_{j,t}) \right] \rightarrow 0 \,\, \text{for} \,\, t \rightarrow +\infty
\end{equation*}
Q.D.E.
\begin{remark}
The same result can be relaxed and hold if $r$ has a distribution stably converging to its mean $\mu_r$ for $t \rightarrow +\infty$
\end{remark}
\begin{remark}
If we introduce heterogeneous initial distribution of wealth (that is, $W_{i,1} \neq W_{j,1}$ for some $i,j$), it can be proved that the simple return dynamics tends to be neutral on the initial ranking for $t \rightarrow +\infty$, that is, the initial ranking is maintained in the long-run under the simple return structure (see also \citealt{biondi2018financial}).
\end{remark}
\end{proof}

\bibliographystyle{spbasic}
 \newcommand{\noop}[1]{}

\end{document}